\documentclass[10pt]{iopart}
\expandafter\let\csname equation*\endcsname\relax
\expandafter\let\csname endequation*\endcsname\relax
\usepackage{amsmath}
\usepackage{amsthm}
\usepackage{iopams}
\usepackage{bbm}
\usepackage{braket}
\usepackage{graphicx}
\usepackage{epstopdf}
\usepackage{enumitem}
\usepackage[colorlinks=true,linkcolor=blue,citecolor=magenta,urlcolor=blue]{hyperref}
\usepackage{hyperref} 
\usepackage[dvipsnames]{xcolor}
\usepackage{soul}
\usepackage{color}
\usepackage{multirow}
\usepackage{cite}
\usepackage{mathtools}
\usepackage{subcaption}
\usepackage{physics}

\setul{}{1.5pt}

\setstcolor{blue}
\newcommand\blk[1]{\color{black}#1}

\newcommand{\ro}[1]{\left( {#1} \right)}

\theoremstyle{plain}
\newtheorem{definition}{Definition}
\newtheorem{theorem}[definition]{Theorem}
\theoremstyle{plain}
\newtheorem{lemma}[definition]{Lemma}
\newtheorem{corollary}[definition]{Corollary}

\usepackage{sidecap}
\usepackage{etoolbox}
\AtBeginEnvironment{figure}{\mathindent=0pt }

\begin{document}

\title{On the power of one pure steered state for EPR-steering with a pair of qubits}

\author{Qiu-Cheng Song$^{1,2}$, Travis J. Baker$^1$  and Howard M. Wiseman$^1$}

\address{$^1$ Centre for Quantum Computation and Communication Technology (Australian Research Council),\\ Centre for Quantum Dynamics, Griffith University, Yuggera Country, Brisbane, Queensland 4111, Australia}
\address{$^2$ School of Physical Sciences, University of Chinese Academy of Sciences, Yuquan Road 19A, Beijing 10049, China}

\ead{songqiucheng19@gmail.com, dr.travis.j.baker@gmail.com, prof.howard.wiseman@gmail.com}
\date{\today}
\vspace{10pt}

\begin{abstract}
As originally introduced, the EPR phenomenon was the ability of one party (Alice) to steer, by her choice between two measurement settings, the quantum system of another party (Bob) into two distinct ensembles of pure states. As later formalized as a quantum information task, EPR-steering can be shown even when the distinct ensembles comprise mixed states, provided they are pure enough and different enough. Consider the scenario where Alice and Bob each have a qubit and Alice performs dichotomic projective measurements. In this case, the states in the ensembles to which she can steer form the surface of an ellipsoid ${\cal E}$ in Bob's Bloch ball. Further, let the steering ellipsoid ${\cal E}$ have nonzero volume (as it must if the qubits are entangled). It has previously been shown that if Alice's first measurement setting yields an ensemble comprising {\em two} pure states, then this, plus any one other measurement setting, will demonstrate EPR-steering.  Here we consider what one can say if the ensemble from Alice's first setting contains 
only {\em one} pure state $\mathsf{p}\in{\cal E}$, occurring with probability $p_\mathsf{p}$. Using projective geometry, we derive the necessary and sufficient condition analytically for Alice to be able to demonstrate EPR-steering of Bob's state using this and some second setting, when the two ensembles from these lie in a given plane. Based on this, we show that, for a given ${\cal E}$, if $p_\mathsf{p}$ is high enough [$p_{\sf p} > p_{\rm max}^{{\cal E}} \in [0,1)$] then {\em any} distinct second setting by Alice is sufficient to demonstrate EPR-steering.  Similarly, we derive a $p_{\rm min}^{{\cal E}}$ such that $p_\mathsf{p}>p_{\rm min}^{{\cal E}}$ is necessary for Alice to demonstrate EPR-steering using only the first setting and {\em some} other setting. Moreover, the expressions we derive are tight; for spherical steering ellipsoids, the bounds coincide: $p_{\rm max}^{{\cal E}} = p_{\rm min}^{{\cal E}}$.  
\end{abstract}

\section{Introduction}

In 1935, Einstein, Podolsky and Rosen (EPR) first pointed out the seemingly non-local behaviour exhibited by two entangled quantum mechanical systems~\cite{EPR35}. 
In response, Schrödinger introduced the terms \emph{entanglement} and \emph{steering} to describe how Alice can influence Bob’s particle from a distance through her choice of measurement~\cite{Sch35}. 
We now know that the phenomenon of steering described by EPR is a distinct kind of non-locality,  strictly between entanglement and Bell non-locality~\cite{Wis07, Jon07}. 
If the assemblage of all Bob's ensembles that could be remotely generated by Alice can be simulated by a local-hidden-state (LHS) model, then the state shared by Alice and Bob is said to be non-steerable from Alice to Bob.
If such a model does not exist, the state is steerable~\cite{Wis07, Jon07}. 

The simplest quantum steering scenario is where Alice and Bob share a pair of qubits. 
Conditioned on Alice's measurement and outcome, each ensemble prepared for Bob contains qubit states, which can be represented as points in the Bloch Ball. 
The set of all qubit states that Bob can be steered to, by permitting Alice to perform all possible measurements on her qubit, form an ellipsoid inside his Bloch ball. 
This geometric object is called the quantum steering ellipsoid~\cite{Jev14}, and is determined by the two-qubit state they share. 
Specifically, the surface of the steering ellipsoid can be obtained from projective measurements, and interior points from general measurements described by positive operator-valued-measures (POVMs)~\cite{Jev14}. 
It is worth noting that not every ellipsoid inside the Bloch
sphere corresponds to a physical two-qubit state. 
This has been analyzed by Milne {\em et al}.~\cite{Mil14}, who derived necessary and sufficient conditions for physicality. 
It is known that a geometric characterization of the quantum steering ellipsoid, together with knowledge of Alice and Bob's Bloch vectors, allows one to faithfully reconstruct any two-qubit state~\cite{Jev14}. 
Independent of the concept of EPR-steering, the steering ellipsoid formalism has many interesting connections to other fields of quantum information theory---for example, its volume is a nonlinear entanglement criterion~\cite{Jev14}. 

It is important to maintain a clear distinction between the steering ellipsoid formalism and \emph{EPR}-steering.
The latter requires determining the existence (or not) of a LHS model for the ensembles of quantum states Bob is steered to by Alice.
When Bob's system is a qubit, the geometry of these steered states is conveniently captured by the steering ellipsoid formalism. 
However, the mere existence of a steering ellipsoid does not necessarily convey information about steerability, since classically correlated (non-entangled, and therefore non-steerable) quantum states can correspond to steering ellipsoids with non-zero volume. 
In general, there is no simple way to detect EPR-steerability from only steering ellipsoids.
One example where such a link was made was in~\cite{Jev15}, where a necessary condition for steering by Alice, when she is permitted to perform all projective measurements, was derived for T-states, which comprise all mixtures of the four two-qubit Bell states.
This condition was derived from  
the ansatz that Bob's distribution of LHSs could be proportional to the fourth power of the distance of the steering ellipsoid surface from the origin. 
This condition was later shown also to be sufficient~\cite{Ngu16a}. 
In this paper, we make another connection between EPR-steering and the steering ellipsoid for a very different class of two-qubit states---those which permit exactly one of Bob's steered states to be pure.

EPR-steering has been studied extensively for two-qubit states (see, e.g., Refs. ~\cite{Wis07,Jon07,Rei09,Uol20,Jev14,Jev15,Ngu16a,Mil14,McC17,Bow16,Ngu16,Shuming16,Yu18,Yubaichu18,Bak18,Bak20,qno02}). 
 The earliest approaches to detecting steering were based on using steering inequalities~\cite{Cav09}, which are analogous to Bell inequalities, and proved useful in experimental demonstrations with entangled qubit pairs~\cite{Sau10}. Alternatively, 
given a finite set of steered ensembles for Bob---called an \emph{assemblage}~\cite{Pus13}---semi-definite programs can be used to determine steerability numerically~\cite{Skrzypczyk14,Hir16,Fil18,Cav16a}; see~\cite{Cav16} for a review.
Another relevant approach, which is applicable to infinitely large assemblages, was developed in Ref.~\cite{Ngu19}.
There, by utilizing the geometry of the two-qubit steering scenario, and bounding the Bloch sphere inside and outside by discrete meshes of points, it was shown that the steerability of almost any two-qubit state can be numerically determined to high accuracy by executing a linear program. 
However, exact necessary and sufficient criteria for steering under all projective measurements are known only for the class of T-states, as mentioned above. 
Finally, of particular relevance to the current paper, for the simple case of two dichotomic measurements, there exist some analytical criteria to characterize EPR-steering under various assumptions on the class of states and measurements; see~\cite{cavalcanti2015analog, PhysRevA.94.032317,PhysRevA.95.062111,quan2016steering,chen2017quantum}.

As our title suggests, here we are interested in assemblages that contain at least one pure state.
In Ref.~\cite{Ngu17}, it was discovered that pure steered states often carry interesting information about the shared bipartite state. 
In particular, if one of Alice's projective measurements can steer Bob's qubit to two pure states, then the bipartite system they share is either separable or EPR-steerable.
For qubits, if these states are identical, then the two-qubit state shared by Alice and Bob is a direct-product state.
If these pure states are different, then the two-qubit state is steerable by performing that measurement, and any other measurement.
These results were a concise generalization of so-called ``all-versus-nothing proofs'' of steering~\cite{Che13, sun14}. 
These results generalize naturally to high dimensional systems.
If Alice can steer Bob's system to a set of independent pure states, then the bipartite state shared by Alice and Bob is either separable or steerable~\cite{Ngu17}.

In this paper we present strong results for the case where Alice can perform two dichotomic projective measurements on her qubit,  and can steer Bob's qubit to exactly {\em one} pure state with some non-zero probability.
This situation can arise with any two qubit state where Bob's steering ellipsoid touches his Bloch sphere.
We find that EPR-steering is determined completely by the location of Bob's reduced state, and the geometry of his steering ellipsoid. 
Remarkably, we find interesting relationships between LHS models, the steering ellipsoid, and projective geometry.
Using the latter, we derive theorems which provide practical and simple criteria for steering, which are both necessary and sufficient. We illustrate the elegant nature of our criteria by explicitly constructing the projective transformations which certify steerability/non-steerability, for three sets of entangled states, some of which have been studied in the entanglement literature.
These are tangent X-states, canonical obese states, and tangent sphere states. Using the same tools, we also show that if the probability of Alice steering Bob’s qubit to a pure state is above a certain threshold, which can be readily computed given his steering ellipsoid, then Alice can always demonstrate EPR-steering by two dichotomic measurements.

This paper is structured as follows.
In Sec.~\ref{sec:one_pure_assemblage} we introduce preliminaries and notations for EPR-steering.
We discuss assemblages in the simplest steering scenario, introducing the term quadrivial assemblage, and a special case of these which contain exactly pure state. 
We finish by deriving conditions for steerability of assemblages of the latter kind.
In Sec.~\ref{sec:Scenario} we introduce two-qubit states represented in terms of the Pauli matrices, and state conditions for one of the steered states to be pure. 
We also rigorously define the scenario we consider in this paper.
In Sec.~\ref{sec:projective_geometry} we introduce relevant background about projective geometry and derive an interesting geometric result, using projective transformations.
We then use this fact to formulate new EPR-steering criteria for all those two-qubit states which permit one pure steered state. 
In Sec.~\ref{sec:examples} we use our main theorems to analytically derive the necessary and sufficient EPR-steering criteria for tangent X-states, canonical obese states and tangent sphere states in the Scenario we defined.  
In Sec.~\ref{sec:one_pure_state_power}, we give a necessary criterion and a sufficient criterion for Alice to demonstrate EPR-steering in our Scenario,  which requires only knowledge of Bob's ellipsoid and the probability that his qubit can be steered to a pure state. 
 Using the main Theorem we derive, we then apply this result to two examples, which are tangent X-states and tangent spheroid states. 
We conclude in Sec.~\ref{sec:Conclusion}, and discuss open questions.

\section{EPR-steering}
\label{sec:one_pure_assemblage}

In this section, first we introduce the concepts of LHS model and EPR-steering. 
Then we give the definition of quadrivial  assemblage and a lemma about a quadrivial assemblage admitting a LHS model. 
Finally, we give the definition of one pure quadrivial assemblage and the necessary and sufficient for one pure quadrivial assemblage to admit a LHS model. In this section, we do not require Alice's system to be a qubit.

\subsection{Ensembles and local-hidden-state models}

Consider a bipartite state $\rho_\mathrm{AB}$ shared by Alice and Bob. 
Suppose Alice can choose to perform different measurements on her subsystem, labelled by $x$, which gives outcomes labelled by $r$.
Generally, the effects of this type of measurement are described by elements of a POVM $E_{r \mid x}$, which satisfy the relations $\sum_{r} E_{r \mid x}=I\ \forall x$ and $E_{r \mid x} \geqslant 0 \ \forall r, x$. 
Conditioned on each measurement setting and outcome that Alice produces, the state of Bob’s system is steered into the state
\begin{equation}
\rho_{r \mid x}=\frac{\operatorname{Tr}_{\mathrm{A}}\left[\left(E_{r \mid x} \otimes {I}\right) \rho_{\mathrm{AB}}\right]}{\operatorname{Tr}\left[\left(E_{r \mid x} \otimes {I}\right) \rho_{\mathrm{AB}}\right]}
\end{equation}
with probability $p_{r|x}=\operatorname{Tr}\left[\left(E_{r \mid x} \otimes {I}\right) \rho_{\mathrm{AB}}\right]$. 
The information available to Bob is the conditional states and corresponding probabilities for all measurements and outcomes on Alice's side, which are ensembles of quantum states $\{\sigma_{r|x}\}_{r,x}$, where $\sigma_{r|x}=p_{r|x}\rho_{r|x}$.
The set of all ensembles arising in a steering scenario is often called an {\em assemblage}~\cite{Pus13}. 
 An important property of assemblages is that they are non-signalling,
\begin{equation}
\sum\limits_{r} p_{r|x}\rho_{r|x} = \sum\limits_{r} p_{r|x'}\rho_{r|x'},
\end{equation}
for all $x,x'$.
This ensures that Bob's marginal state is independent of Alice's measurement choice.
Importantly, the assemblage contains all relevant information for a steering test.
We say that the assemblage $\{\sigma_{r|x}\}_{r,x}$ is non-steerable if and only if it can be reproduced by a coarse graining over local quantum states on Bob's side~\cite{Wis07}.
That is, if and only if, for all $r$ and $x$, it admits a decomposition
\begin{align}\label{LHS}
p_{r|x} \rho_{r|x} = \sum_{\lambda}  D(r|x, \lambda) p_\lambda \rho_{\lambda},
\end{align}
where $\{ p_\lambda, \rho_\lambda \}$ is the ensemble of LHSs, and $D(r|x, \lambda)$ is a set of probability distributions labelled by $\lambda$ which map measurement settings to outcomes $a$ deterministically~\cite{Pus13}. 
Conversely, an assemblage demonstrates steering if and only if it does not admit a decomposition as in Eq.~\eqref{LHS}.

\subsection{Simplest steering: the quadrivial assemblage}

In order to study the power of assemblages containing one (normalized) pure state, we first focus on the simplest family of assemblages capable of demonstrating steering~\cite{Sau12}. 
This is the family of assemblages arising from two measurements, with two outcomes.
We refer to this as a quadrivial assemblage. 
The name reflects the fact that there are four states in the assemblage, and that this is the smallest non-trivial assemblage, in the sense that it has the potential to demonstrate EPR-steering. 

\begin{definition}[Quadrivial assemblage]
Alice can perform two dichotomic measurements on her local system to steer Bob's system into two distinct ensembles of states, $\{p_{\pm|x}, \rho_{\pm | x} \}_x$ with $x=0,1$. 
We call the assemblage $\{\sigma_{r|x}\}_{r,x}$, where $\sigma_{r|x}=p_{r|x}\rho_{r|x}$, a quadrivial assemblage.
\end{definition}
From Eq.~\eqref{LHS}, the quadrivial  assemblage admits a LHS model if and only if there exists a valid solution to the system of equations
\begin{align}
&p_{+|0} \rho_{+ | 0} = p_0\rho_{0} + p_1\rho_{1}, \label{lhsq1}\\
&p_{-|0} \rho_{- | 0} = p_2\rho_{2} + p_3\rho_{3},\label{lhsq2}\\
&p_{+|1} \rho_{+|1} = p_0\rho_{0} + p_2\rho_{2},\label{lhsq3}\\
&p_{-|1} \rho_{-|1} = p_1\rho_{1} + p_3\rho_{3}. \label{lhsq4}
\end{align}
Now we will focus again on the simplest case, which is where Bob's system is of dimension two.
This means that his steered states will be bounded operators on $\mathcal{H}^2$ with unit trace.
These can be represented in the Bloch ball $\mathcal{B}$, wherein a point on its surface (the Bloch sphere $\mathcal{S}$) represents a pure qubit state, and points interior correspond to mixed states~\cite{Nie02}.
We will see that determining the steerability of qubit-quadrivial assemblages can be determined elegantly from their geometric properties.
For this purpose, we will represent quantum states appearing in qubit assemblages by points inside the Bloch ball, without reference to an origin. 
We use $\mathsf{s}_{r|x}$ (resp. $\mathsf{s}_{\lambda}$) to denote the points representing the quantum states of the same label, $\rho_{r|x}$ ($\rho_{\lambda}$).
From a geometrical point of view, mixtures of two quantum states correspond to all points on a straight line connecting the states inside the Bloch ball $\mathcal{B}$~\cite{Ben06}.
For example, Eq.~\eqref{lhsq1} shows that two states $\rho_0$ and $\rho_1$ are mixed with probabilities $p_0/p_{+|0}$ and $p_1/p_{+|0}$ respectively, averaging to the state $\rho_{+\mid 0}$. 
This means that the point $\mathsf{s_{+\mid 0}}$, which corresponds to $\rho_{+\mid 0}$, must be on the segment of the straight line between the points $\mathsf{s_{0}}$ (corresponding to $\rho_0$) and $\mathsf{s}_{1}$ (corresponding to state $\rho_1$) in Bob's Bloch ball. 
The constraints that a LHS ensemble must satisfy can be straightforwardly expressed in terms of constraints on probabilities and geometric points by
\begin{align}
p_{r|x} &= \sum_{\lambda}  D(r|x, \lambda) p_\lambda \label{eq:points_probabilities} \\
p_{r|x} \mathsf{s}_{r|x} &= \sum_{\lambda}  D(r|x, \lambda) p_\lambda \mathsf{s}_{\lambda}, \label{eq:points}
\end{align}
for all $r,x$.
Likewise, the no-signalling condition becomes,
\begin{equation}
\sum\limits_{r} p_{r|x}\mathsf{s}_{r|x} = \mathsf{b}\quad\forall x. \label{eq:NSpoints}
\end{equation}
We note that the necessary and sufficient condition for the quadrivial assemblage to exhibit steering can also be derived from a steering inequality~\cite{PhysRevA.95.062111,PhysRevLett.115.230402} derived in Ref.~\cite{PhysRevA.81.062116}.

\subsection{One pure quadrivial assemblage}

Now, we make a restriction to consider only quadrivial assemblages that contain exactly one pure steered state. 
We will refer to these as 1PQAs. 
\begin{definition}[1-pure quadrivial assemblage (1PQA)]
\label{definition2}
 A 1-pure quadrivial assemblage (1PQA) is a quadrivial assemblage that contains 
 exactly one pure steered state with nonzero probability.
\end{definition}

In what follows, we continue to restrict Bob's system to be a qubit.
Note that we do not yet make any assumptions about Alice's system. 
If the steered states for a 1PQA are not collinear, then the steered states must lie in a unique plane ${P}_\theta$ which intersects the Bloch ball $\mathcal{B}$.
One such plane is depicted in Fig.~\ref{triangle}, where the circle $\mathsf{C}_\theta$ is the boundary of a disk formed by the intersection of $\mathcal{B}$ and ${P}_\theta$.
There, the steered states (green points) are reproduced by averaging over members of the LHS ensemble, shown as the black points which form the vertices of a triangle.
Interestingly, the existence of the pure state in 1PQAs implies that steerability is determined entirely by whether one such triangle exists inside $\mathsf{C}_\theta$.
This is expressed in the following lemma.

\begin{figure}[h]\centering
\includegraphics[angle=0,width=0.42\linewidth]{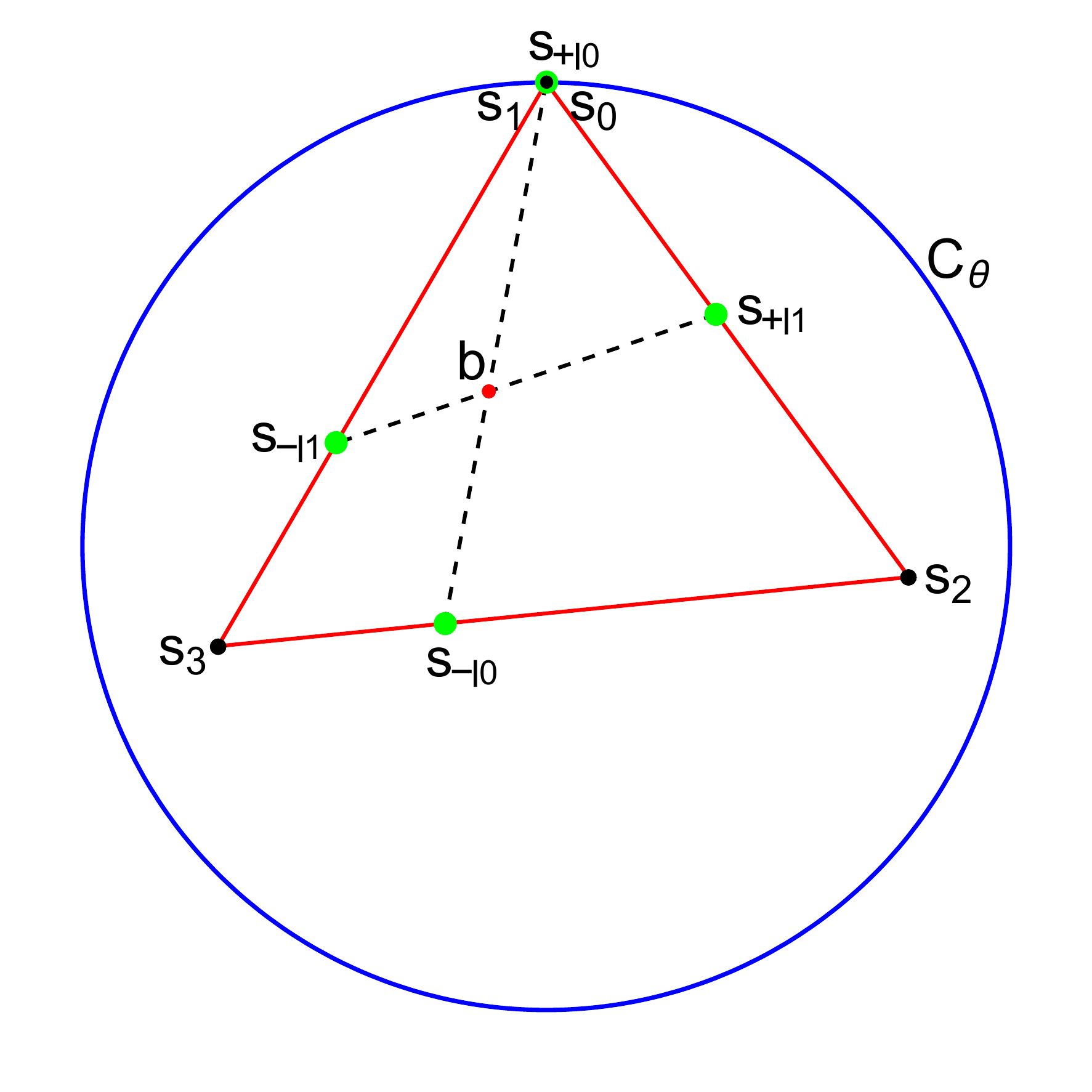}
\caption{Depiction of a generic 1PQA which admits an LHS model in a plane intersecting the Bloch ball.
The red point $\mathsf{b}$ represents Bob's reduced state, and the four green points  $\mathsf{s}_{\pm|x}$ are the normalized members of the 1PQA. 
The LHS ensemble is composed of four quantum states at points $\mathsf{s}_{i}$.
} 
\label{triangle} 
\end{figure}
\begin{lemma}\label{lemma2}
 A qubit-1PQA admits a LHS model if and only if there exists a triangle in the Bloch ball that contains all the steered states on its boundary.
\end{lemma}

\begin{proof}
By suitable labelling of the assemblage, we can make $\rho_{+|0}$ the pure state which is required in the definition of the 1PQA. 
Since pure quantum states are extremal points in the set of quantum states, they cannot be expressed as a mixture of two \emph{different} states.

To prove the \emph{if} part, we are required to show that any triangle which satisfies the requirements of the lemma corresponds to a valid LHS model for the 1PQA.
Since $\mathsf{s}_{+|0}$ is located on the boundary $\mathsf{C}_\theta$, any triangle contained in the Bloch Ball must have one of its vertices located at the same point.
For the corresponding LHS model, we make the choice $\mathsf{s}_{0}=\mathsf{s}_{1}=\mathsf{s}_{+\mid 0}$.
Now, since the remaining three steered states are not collinear, any triangle satisfying the conditions of the lemma must also have exactly one of $\mathsf{s}_{-\mid 0}, \mathsf{s}_{+\mid 1}, \mathsf{s}_{-\mid 1}$ along each of its edges.
Any triangle with this property can be described by weights  $\epsilon_{-|0} \in \left(0,1\right)$ and $\epsilon_{\pm |1} \in \left[0,1\right)$ such that
\begin{align}
\mathsf{s}_{- \mid 0} &= \epsilon_{-|0} \mathsf{s}_{2}+(1-\epsilon_{-|0})\mathsf{s}_{3}, \label{tria1}\\
\mathsf{s}_{+ \mid 1} &= \epsilon_{+ |1} \mathsf{s}_{0}+(1-\epsilon_{+ |1})\mathsf{s}_{2}, \label{tria2}\\
\mathsf{s}_{- \mid 1} &= \epsilon_{-|1} \mathsf{s}_{1}+(1-\epsilon_{-|1})\mathsf{s}_{3}. \label{tria3}
\end{align}

In other words, $\mathsf{s}_{2}$ and $\mathsf{s}_{3}$ are the other two vertices of the triangle $\triangle \mathsf{s}_0 \mathsf{s}_{2} \mathsf{s}_{3}\equiv \triangle \mathsf{s}_1 \mathsf{s}_{2} \mathsf{s}_{3}$.
Now for any such arrangement of the $\mathsf{s}_i$'s,  the probabilities with which each LHS appears are uniquely determined by the constraints on the ensemble which does \emph{not} contain the pure state.
That is, for $x=1$, Eqs.~\eqref{eq:points_probabilities} and \eqref{eq:points} require that
\begin{align}
p_0 &= p_{+|1}\epsilon_{+ |1},\label{lhsp0} \\
p_1 &= p_{-|1}\epsilon_{-|1}, \label{lhsp1} \\
p_2 &= p_{+|1}(1-\epsilon_{+ |1}),\label{lhsp2} \\
p_3 &= p_{-|1}(1-\epsilon_{-|1}). \label{lhsp3}
\end{align}
The probabilities in Eqs.~\eqref{lhsp0}-\eqref{lhsp3} and the positions of the $\mathsf{s}_i$'s entirely define a LHS model, characteristic to the triangle they define.
It remains to check that the $x=0$ ensemble is reproduced by it.
The mixed steered state in this ensemble is required to satisfy
\begin{equation}
p_{-|0} \mathsf{s}_{-|0} = p_2 \mathsf{s}_{2} + p_3 \mathsf{s}_{3},
\end{equation}
which can be equivalently expressed in terms of the triangle vertices as
\begin{equation}
p_{-|0} (\epsilon_{-|0} \mathsf{s}_{2}+(1-\epsilon_{-|0})\mathsf{s}_{3}) = p_2 \mathsf{s}_{2} + p_3 \mathsf{s}_{3}. \label{eq:mixed_state_vertex}
\end{equation}
To verify this holds, we use the no-signalling condition in Eq.~\eqref{eq:NSpoints} twice; once for each measurement.
That is, $\mathsf{b} = p_{+|0} \mathsf{s}_{+|0} + p_{-|0} \mathsf{s}_{-|0} = p_{+|1} \mathsf{s}_{+|1} + p_{-|1} \mathsf{s}_{-|1}$.
Expressing these in terms of the triangle vertices, we have
\begin{align}
\mathsf{b}
&= p_{+|0}\mathsf{s}_0 + p_{-|0} \epsilon_{- |0} \mathsf{s}_2 + p_{-|0}(1-\epsilon_{-|0}) \mathsf{s}_3 ,\label{bc1} \\
&=(p_{+|1}\epsilon_{+ |1}+p_{-|1}\epsilon_{-|1}) \mathsf{s}_0 + p_{+|1}(1-\epsilon_{+ |1})\mathsf{s}_2 +p_{-|1}(1-\epsilon_{-|1})\mathsf{s}_3.
\end{align}
These two representations are equivalent, since they are both normalized barycentric coordinate representations of $\mathsf{b}$ with respect to the triangle $\triangle \mathsf{s}_0 \mathsf{s}_{2} \mathsf{s}_{3}$, and such a representation is unique~\cite{ungar10}.
Therefore, by matching terms, we can deduce that
\begin{align}
& p_{+|0}=p_{+|1}\epsilon_{+ |1}+p_{-|1}\epsilon_{-|1},\\
& p_{+|1}(1-\epsilon_{+ |1}) = p_{-|0} \epsilon_{- |0} \equiv p_2, \label{lhsp22} \\
& p_{-|1}(1-\epsilon_{-|1}) = p_{-|0}(1-\epsilon_{-|0}) \equiv p_3,\label{lhsp33}
\end{align}
which means that Eq.~\eqref{eq:mixed_state_vertex} is satisfied.
From these three equations, we can also verify the states in the $x=0$ ensemble appear with the right probabilities,  $p_{+|0} = p_0 + p_1$ and $p_{-|0} = p_2 + p_3$.
This proves the \emph{if} part, since we have shown that any triangle satisfying the requirements of the lemma corresponds to a valid LHS model for a 1PQA.

The \emph{only if} part of the statement is straightforward, since if there is no triangle in the Bloch ball, then at least one of the points $\mathsf{s}_i$ required by \eqref{eq:points} to reproduce the steered ensembles must be outside it, corresponding to a non-quantum state.
\end{proof}

While Lemma~\ref{lemma2} permits a purely geometric characterization of steerability, it requires checking the existence for \emph{all} possible triangles passing through the steered states. 
The next observation allows a simpler characterization by considering only a single triangle.

\begin{figure}[h]\centering
\includegraphics[angle=0,width=0.42\linewidth]{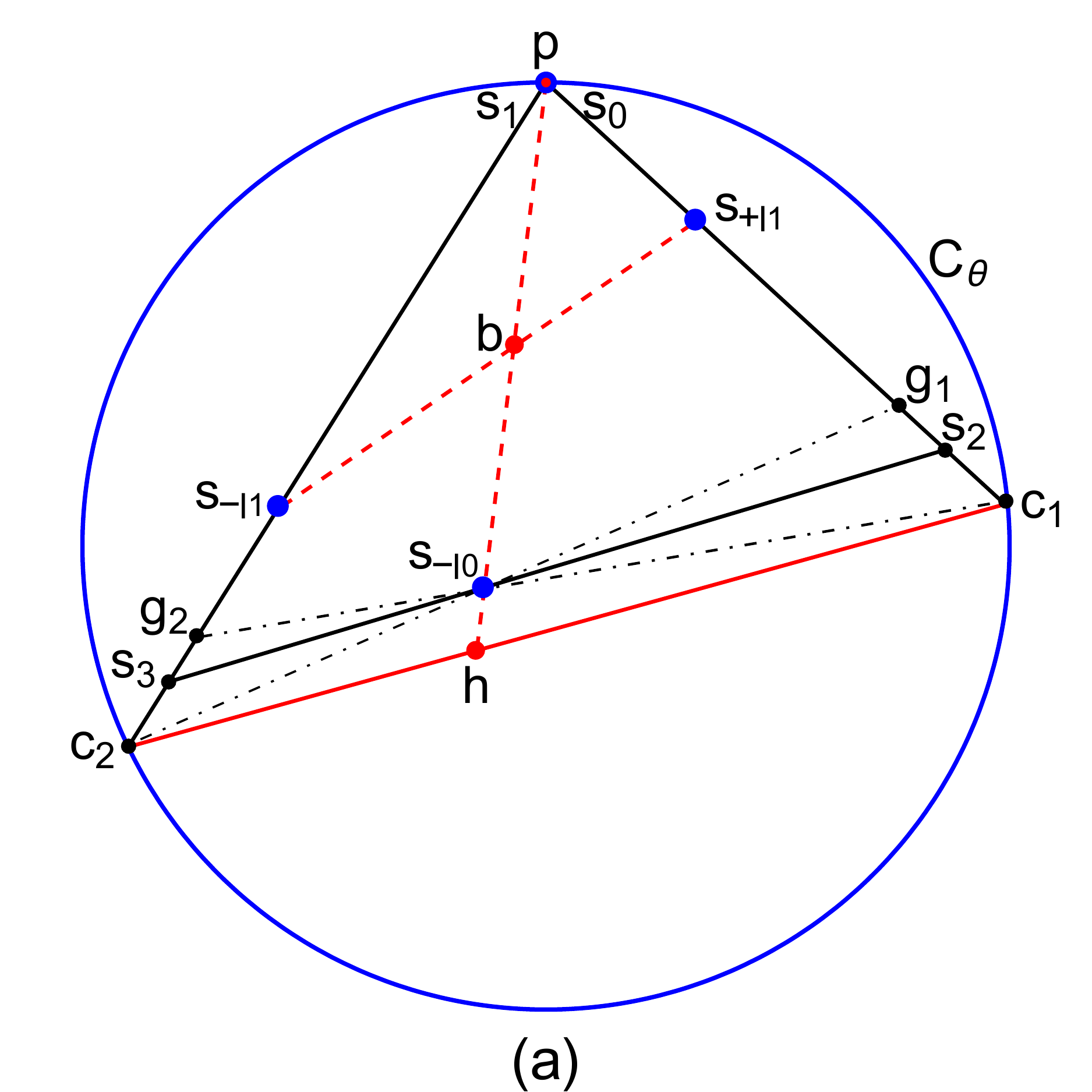}
\hspace{7pt}
\includegraphics[angle=0,width=0.42\linewidth]{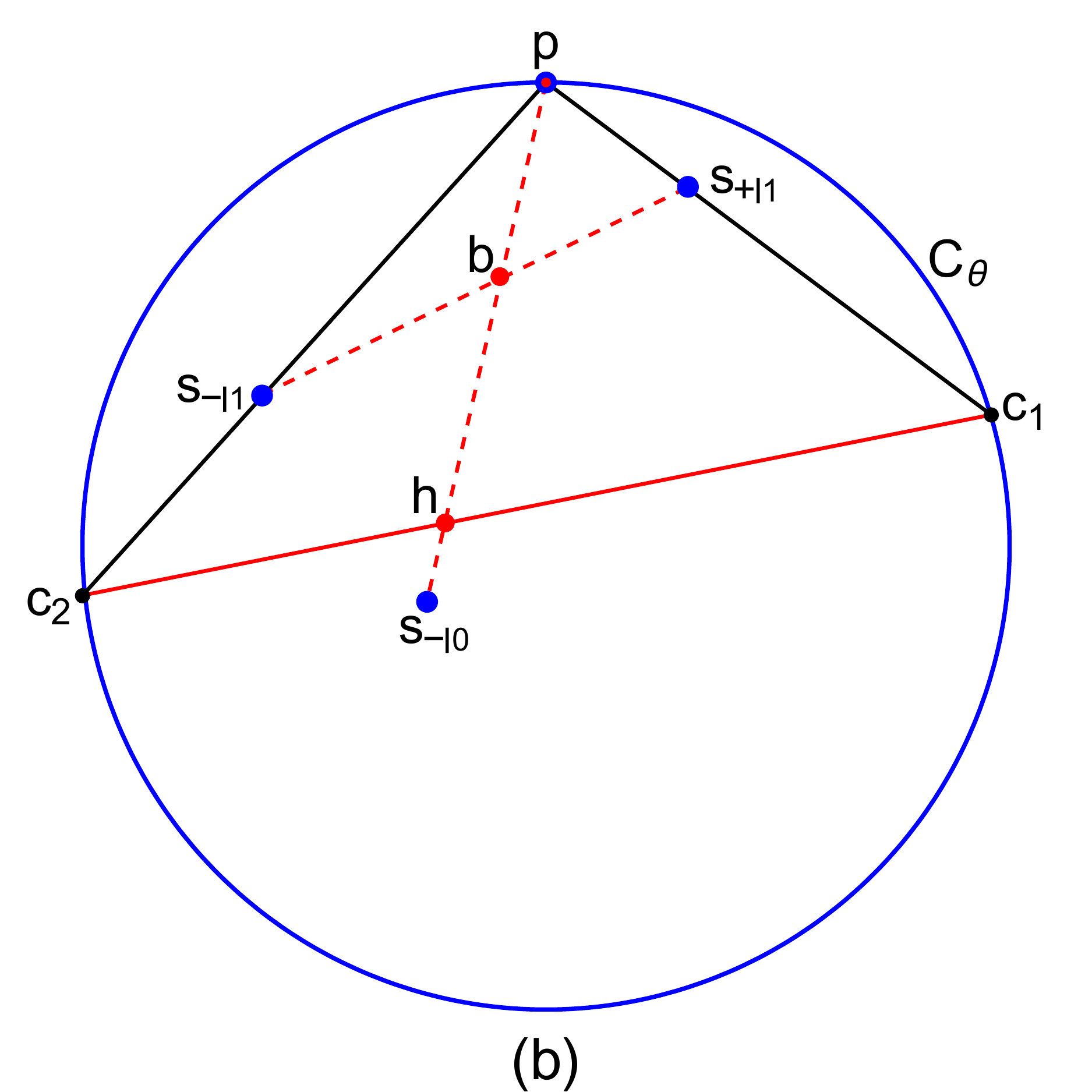}
\caption{Geometry relevant to Corollary \ref{corollary}. 
Bob's steered ensembles consist of four steered states represented by  $\mathsf{s}_{\pm|x}$.
The blue circle $\mathsf{C_\theta}$ is the cross section between the plane defined by these steered states and the Bloch sphere.
The point $\mathsf{p}$ represents the pure state, and the red point $\mathsf{b}$ represents Bob's reduced state. 
Corollary \ref{corollary} shows that us that a qubit-1PQA demonstrates EPR-steering if and only if $\mathsf{s_{-|0}}$ is outside the triangle $\triangle \mathsf{pc_{1}c_{2}}$. 
(a) A non-steering case. 
The LHS $\mathsf{s_2}$ (resp. $\mathsf{s_3}$) can be located at any point on the line segment $\mathsf{g_1c_1}~(\mathsf{g_2c_2})$ to form $\triangle\mathsf{ps_2s_3}$, which satisfies Lemma~\ref{lemma2}. 
(b) A steering case.
\label{geo}
}
\end{figure}

\begin{corollary}\label{corollary}
Consider a qubit-1PQA contained in the plane ${P}_\theta$.
Define $\mathsf{c}_1$ and $\mathsf{c}_2$ as the intersection of $\mathsf{C}_\theta = {P}_\theta \cap \mathcal{S}$ with the extension of the lines through the pure steered state $\mathsf{p}$ and the two states from the \emph{other} steered ensemble.  
The 1PQA is non-steerable if and only if the mixed state in the same ensemble as the pure state is contained within the convex hull of $\triangle \mathsf{p}\mathsf{c}_1 \mathsf{c}_2$. 
\end{corollary}

\begin{proof}
As in the proof of Lemma~\ref{lemma2}, we assume that the pure steered state $\mathsf{p}$ corresponds to $\mathsf{s}_{+\mid 0}$; see Fig.~\ref{geo}.
To prove sufficiency, suppose $\triangle \mathsf{p}\mathsf{c}_1 \mathsf{c}_2$ contains the point $\mathsf{s}_{-\mid 0}$, which corresponds to the mixed state in the same ensemble as $\mathsf{p}$.
By construction, $\mathsf{c}_1$ and $\mathsf{c}_2$ are collinear with the pairs $(\mathsf{p},\mathsf{s}_{+|1})$ and $(\mathsf{p},\mathsf{s}_{-|1})$ respectively, they can always be moved along these lines toward (but not beyond) $\mathsf{s}_{+|1}$ and $\mathsf{s}_{-|1}$ to find a triangle which satisfies Lemma~\ref{lemma2}.

Necessity follows from the fact that $\mathsf{c}_1$ and $\mathsf{c}_2$ are on extreme points of the Bloch ball.
Since $\mathsf{s}_{+|1}$ and $\mathsf{s}_{-|1}$ must necessarily lie on the sides of a triangle for a LHS model to exist, its two vertices not at $\mathsf{p}$ can only be as far away as $\mathsf{c}_1$ and $\mathsf{c}_2$.
If $\mathsf{s}_{-\mid 0}$ is outside $\triangle \mathsf{p}\mathsf{c}_1 \mathsf{c}_2$, at least one of the vertices of the triangle must lie outside $\mathsf{C}_\theta$.
This requirement precludes a LHS model from existing.
\end{proof}

The results in this section apply to any 1PQA, and do not assume anything about the dimension of Alice's system.
In the following sections, we restrict her system to qubits, and prove interesting connections between the geometry of two-qubit steering---as exemplified by the steering ellipsoid formalism---and how it relates to the steerability of 1PQAs.
\blk

\section{The Scenario}
\label{sec:Scenario}

\subsection{The quantum steering ellipsoid}

 Let $\boldsymbol{\sigma}=(\sigma_x,\sigma_y,\sigma_z)^\top$ denote the vector of the Pauli spin operators.
Any two-qubit state $\rho_\mathrm{AB}$ can be expressed in terms of these and the identity operator $I$ as 
\begin{equation}\label{atq}
\rho_\mathrm{AB}=\frac{1}{4}\left({I} \otimes {I}+{\boldsymbol{a}} \cdot \boldsymbol{\sigma} \otimes {I}+{I}\otimes {\boldsymbol{b}} \cdot \boldsymbol{\sigma}+ \boldsymbol{\sigma} \otimes \cdot\ T \boldsymbol{\sigma}
\right),
\end{equation}
 where the combined tensor and dot product are defined, for some vectors of operators ${\boldsymbol{\alpha}}$ and ${\boldsymbol{\beta}}$, as 
${\boldsymbol{\alpha}} \otimes \cdot\ {\boldsymbol{\beta}} = \sum_{i} \alpha_i \otimes \beta_i$.
The parameters in the state are  ${\boldsymbol{a}}=({a}_x,{a}_y,{a}_z)^\top$ and ${\boldsymbol{b}}=({b}_x,{b}_y,{b}_z)^\top$, the Bloch vectors for Alice and Bob’s reduced qubit states, and ${T}$, the spin correlation matrix. 
That is
\begin{align}
&{a}_{i} \coloneqq\operatorname{Tr}\left[\rho_\mathrm{AB}(\sigma_{i} \otimes I)\right], \
{b}_{j}\coloneqq\operatorname{Tr}\left[\rho_\mathrm{AB}(I \otimes \sigma_{j})\right], \
{T}_{ij} \coloneqq\operatorname{Tr}\left[\rho_\mathrm{AB}(\sigma_{i} \otimes \sigma_{j})\right]. \notag
\ 
\end{align}
We are interested in all possible 1PQAs that may arise for Bob as a result of two dichotomic measurements on Alice's side.

If both of Alice's measurements are projective, we denote their effects by 
\begin{equation}\label{projector}
\Pi_{r|x}=\frac{1}{2}(I+r\hat{\boldsymbol{n}}_{x} \cdot \boldsymbol{\sigma})
\end{equation}
with associated unit vector $\hat{\boldsymbol{n}}_{x}$.
The corresponding ensemble for Bob comprises the two conditioned states 
\begin{equation}\label{con-state}
\rho_{r|x}=\frac{1}{2}\left[{I}+\left(\frac{\boldsymbol{b}+r{T}^{\top} \hat{\boldsymbol{n}}_{x}}{1+r\hat{\boldsymbol{n}}_{x}\cdot \boldsymbol{a}}\right) \cdot \boldsymbol{\sigma}\right],
\end{equation}
with corresponding steering probabilities  
$p(r|x)=\frac{1}{2}(1+r\hat{\boldsymbol{n}}_{x} \cdot \boldsymbol{a})$.
Considering all possible projective measurements for Alice, the set of Bloch vectors of Bob’s conditioned states forms the surface of an ellipsoid $\mathcal{E}$ inside his Bloch sphere $\mathcal{S}$:
\begin{equation}\label{qellipsoid}
\partial\mathcal{E}(\boldsymbol{a},\boldsymbol{b},T)=\left\{\frac{\boldsymbol{b}+r T^{\top} \hat{\boldsymbol{n}}_{x}}{1+r\hat{\boldsymbol{n}}_{x} \cdot \boldsymbol{a}}\right\}_{\hat{\boldsymbol{n}}_{x}}
\end{equation}
 Moreover, the interior points of the ellipsoid $\mathcal{E}$ can be steered to if Alice performs POVMs~\cite{Jev14}. The centre of the ellipsoid (not to be confused with Bob's reduced state, having Bloch vector $\boldsymbol{b}$) is
\begin{equation}
\boldsymbol{c}=\frac{\boldsymbol{b}-T^{\top} \boldsymbol{a}}{1-|\boldsymbol{a}|^{2}},
\label{centre}
\end{equation}
and the three semiaxes are the roots of the eigenvalues of the orientation matrix
\begin{equation}\label{orientationmatrix}
Q=\frac{1}{1-|\boldsymbol{a}|^{2}}\left(T-\boldsymbol{a b}^{\top}\right)^{\top}\left(I+\frac{\boldsymbol{a a}^{\top}}{1-|\boldsymbol{a}|^{2}}\right)\left(T-\boldsymbol{a} \boldsymbol{b}^{\top}\right).
\end{equation}
The eigenvectors of $Q$ give the orientation of the ellipsoid around its centre~\cite{Jev14}.

\subsection{Tangent quantum steering ellipsoids}

We are interested in all possible steering ellipsoids which correspond to situations where only one of Bob's steered ensembles contains a pure state. 
From the set of Bloch vectors in Eq.~\eqref{qellipsoid}, it is clear that this will be true if and only if, 
\begin{equation}
\|\boldsymbol{b}+rT^{\top} \hat{\boldsymbol{n}}_{x}\| = 1+r\hat{\boldsymbol{n}}_{x} \cdot \boldsymbol{a}
\end{equation} 
for a unique $\hat{\boldsymbol{n}}_{x}$. 
This ensures that the quantum steering ellipsoid is tangent at exactly one point with Bob's Bloch sphere.
We refer to steering ellipsoids which have this property as \emph{tangent quantum steering ellipsoids.}
By performing a rigid rotation of both the steering ellipsoid and the Bloch ball, we can make that pure steered state $\ket{0}$.
 This is algebraically equivalent to constraining one of the Bloch vectors in Eq.~\eqref{qellipsoid} to satisfy 
\begin{align}\label{onepure1}
\frac{\boldsymbol{b}+rT^{\top} \hat{\boldsymbol{n}}_{x}}{1+r\hat{\boldsymbol{n}}_{x} \cdot \boldsymbol{a}}=(0,0,1)^\top,
\end{align}
 for some choice of $r$ and $x$.
 We will give the corresponding algebraic conditions this requirement places on the steering ellipsoid $\mathcal{E}$ in the next section.

Based on the one pure steered state condition Eq.~\eqref{onepure1}, we define the scenario we will consider.
\begin{definition}[The Scenario]
\label{scenario}
Let Alice and Bob share a pair of qubits such that Bob's qubit's steering ellipsoid $\mathcal{E}$ is of nonzero volume. 
Alice can perform two dichotomic projective measurements on her qubit, and for exactly one outcome of one measurement, steers Bob's qubit to a pure state with some probability $0<p_\mathsf{p}<1$.
\end{definition} 
It is important to mention that a quantum steering ellipsoid with a nonzero volume in the Scenario indicates that the two-qubit state is entangled. In other words, if a pure steered state exists, the two-qubit state is entangled if and only if its steering ellipsoid has a nonzero volume. This is due to the fact that a two-qubit state is separable if and only if its steering ellipsoid fits inside a tetrahedron that fits within the Bloch sphere. 

\section{ Steerability from projective geometry}
\label{sec:projective_geometry}

The Scenario just defined (Def.~\ref{scenario}), gives rise to a 1PQA (Def.~\ref{definition2}), the states of which must, from the steering ellipsoid formalism, 
lie on the boundary of an ellipse. 
Furthermore, the steerability of a 1PQA is determined entirely by the properties of a particular triangle with all of its vertices on $\mathsf{C}_\theta$ ({\em cf.} Corollary~\ref{corollary}), with its shape determined by the geometry of the steered ensembles.
Guided by these observations, we will use tools from projective geometry---which permits transformations between conic sections that preserve collinearities between points---to derive novel conditions for steerability.
In fact, we will see this allows us to determine steerability from steering ellipsoid geometry and the probability with which Bob's pure state appears in his ensemble.
We begin by introducing some basic ideas required for our analysis, and refer the reader to~\cite{Ric11}  
for further background.

\subsection{Projective geometry}

The real projective plane $\mathbb{P}^2$ can be represented by the set of all lines through the origin in $\mathbb{R}^3$.
More precisely, any vector $(x,y,z)^\top$ in $\mathbb{R}^3$ (excluding the zero vector) defines a unique line through the origin.
All scalar multiples of this vector $(\lambda x,\lambda y,\lambda z)^\top$ for $\lambda \neq 0$ form an equivalence class, since they correspond to the same line through the origin (and hence the same point in the projective plane).
Naturally, this leads to a representation of points in the projective plane by vectors in 3-dimensional space using \emph{homogeneous coordinates}.

Our main idea is to take $\mathsf{E_\theta}$ and $\mathsf{C_\theta}$ which exist in the Euclidean plane $\mathbb{R}^2$, and projectively transform them by embedding them in the real projective plane $\mathbb{P}^2$.
In $\mathbb{R}^3$, this can be achieved by taking the Euclidean plane to be the $z=1$ plane~\cite{Ric11}; see Fig.~\ref{fig-map4}(b).
Each point $(x,y)\in\mathbb{R}^2$ then corresponds to the point $(x,y,1)^\top$.
For all other vectors $(x,y,z)^\top$, if $z\neq0$ they correspond to the Euclidean point $(x/z,y/z,1)^\top$; these two sets of homogeneous coordinates in 3-dimensions correspond to the same point in the projective plane.
Vectors with zero $z$ component correspond to ``points at infinity.''
If we embed $\mathsf{C}_\theta$ and $\mathsf{E}_\theta$  into the $z=1$ plane, as shown in Fig.~\ref{fig-map4} with the $\theta$ subscripts omitted, the equations describing them can be generally represented in homogeneous coordinates $\vec{\mathsf{x}}=(x,y,1)^\top$ by all points which satisfy
\begin{equation}
\vec{\mathsf{x}}^\top A \vec{\mathsf{x}} = 0.
\label{eq:conic}
\end{equation}
This form of the equation is called a conic form, and the matrix $A$ can be taken to be real and symmetric.
Ellipses and circles are examples of nondegenerate conics; these are conics which are nonempty and neither a point nor a pair of lines.

For our analysis, we will be interested in manipulating conics by performing projective transformations.
Such transformations can be represented by matrix multiplication in homogeneous coordinates, by a $3\times3$ invertible matrix $H$.
That is, the projective point $\vec{\mathsf{x}}$ is transformed to $H\vec{\mathsf{x}}$. 
If we begin with the conic in Eq.~\eqref{eq:conic}, after the transformation these points satisfy $(H^{-1}\vec{\mathsf{x}})^\top A (H^{-1}\vec{\mathsf{x}})=0$. 
This means that the matrix defining the conic transforms as
\begin{equation}
A \mapsto (H^{-1})^\top A H^{-1}.
\end{equation}
Later, we will use a key property of projective transformations, obvious from the above, that they preserve collinearity (see Theorem 3.2 of~\cite{Ric11}).

Using homogeneous coordinates, we can also give the corresponding algebraic conditions for when ${\cal S}$ touches ${\cal E}$ at exactly one point. 
Consider Bob's Bloch sphere ${\cal S}$:  $\vec{\mathsf{x}}^{\top}\!S \vec{\mathsf{x}}=0$ and the steering ellipsoid ${\cal E}$: $\vec{\mathsf{x}}^{\top}\!E \vec{\mathsf{x}}=0$, where $\vec{\mathsf{x}}=(wx, wy, wz, w)^{\top}$ is the homogeneous coordinate of the point $(x, y, z)$ in $\mathbb{R}^3$, with $w$ being any non-zero real number.
Further, $S$ and $E$ are $4\times4$ real symmetric matrices. 
The characteristic polynomial of the sphere ${\cal S}$ and ellipsoid ${\cal E}$ is defined as
\begin{equation}
\chi(\kappa)=\operatorname{det}( \kappa S-E),
\end{equation}
which is a polynomial of degree 4 in $\kappa$.
The steering ellipsoid ${\cal E}$ is in Bob's Bloch sphere ${\cal S}$ and touches  the surface of the Bloch sphere ${\cal S}$ at a single point if and only if the ordered roots of the characteristic equation $\chi(\kappa)=0$ are real and satisfy \cite{Jia20}
\begin{equation}\label{onepure2}
0<\kappa_1=\kappa_2\leq \kappa_3\leq \kappa_4.
\end{equation}
Having introduced these concepts, we now prove a lemma that will be relevant to steerability in the Scenario (Def.~\ref{scenario}), as we will see in the next subsection.

\begin{figure}[h]\centering
\includegraphics[angle=0,width=0.42\linewidth]{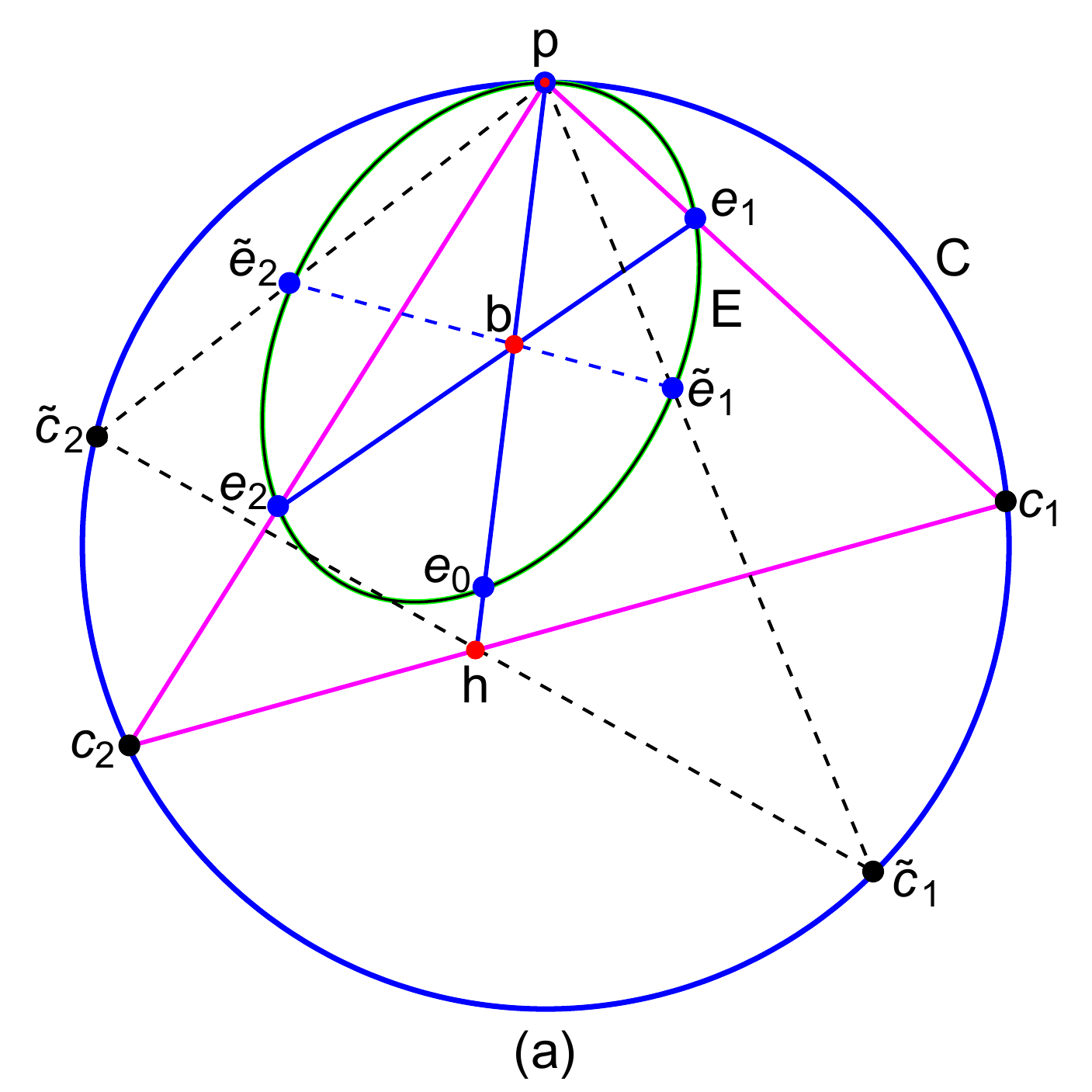}
\hspace{0.2cm}
\includegraphics[angle=0,width=0.42\linewidth]{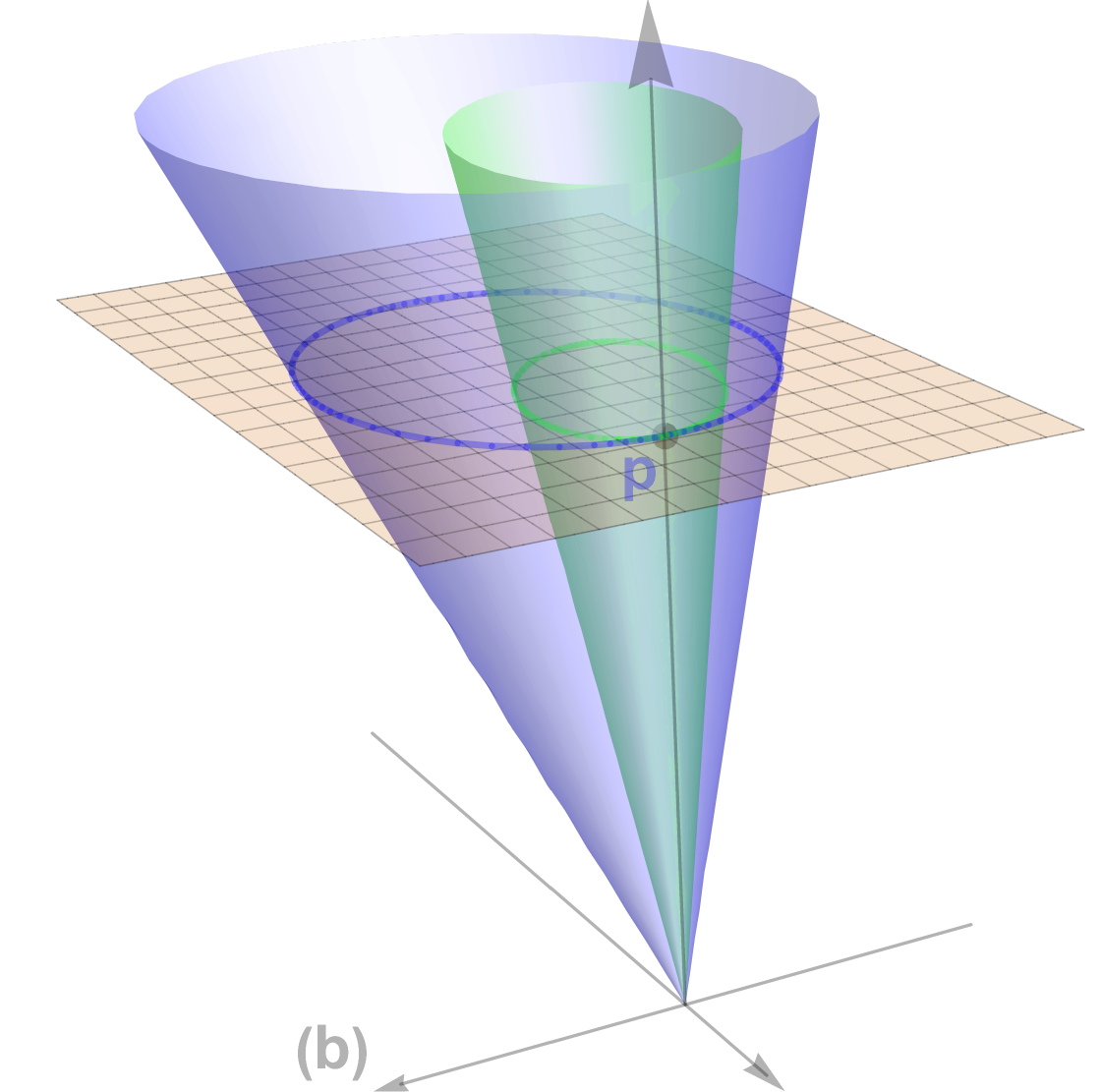}
\vspace{-10pt}
\caption{(a) Geometric structure that arises in consideration of the Scenario, involving a circle $\mathsf{C}$ having an interior tangent ellipse  $\mathsf{E}$, the interior of which contains a given point $\mathsf{b}$. Lemma~\ref{lemma3} shows that the point $\mathsf{h}$ is independent of the choice of $\mathsf{e_1}$ (as illustrated with another choice, $\mathsf{\tilde{e}_1}$). The proof is done using projective geometry. 
(b) Embedding the conics $\mathsf{C}$ and $\mathsf{E}$ in $\mathbb{R}^3$ via the $z=1$ plane (yellow grid). There are infinitely many linear transformations in $\mathbb{R}^3$ mapping the green to the blue cone, only one of which maps an arbitrary $\mathsf{e_1}$ to the $\mathsf{c_1}$ indicated by the construction in (a). 
}
\label{fig-map4} 
\end{figure}

\begin{lemma}
\label{lemma3}
Consider an ellipse $\mathsf{E}$ strictly within a circle $\mathsf{C}$ except for a tangent point $\mathsf{p}$, and a point $\mathsf{b}$ strictly within $\mathsf{E}$, as shown in Fig.~\ref{fig-map4}.
Consider an arbitrary point $\mathsf{e_1}$ on the boundary of $\mathsf{E}$, and a point $\mathsf{e_2}$ given by the intersection of line $\mathsf{be_1}$ with the boundary of $\mathsf{E}$, such that both $\mathsf{e_1}$ and $\mathsf{e_2}$  are distinct from $\mathsf{p}$. 
Define the points $\mathsf{c_1}$ and $\mathsf{c_2}$ as the intersections of  $\mathsf{C}$ with the lines $\mathsf{pe_1}$ and $\mathsf{pe_2}$, respectively. 
Define the point $\mathsf{h}$ as the intersection of $\mathsf{pb}$ and $\mathsf{c_1c_2}$. 
Then the point $\mathsf{h}$ is independent of the choice of $\mathsf{e_1}$.
In particular, $\mathsf{h}$ is the image of $\mathsf{b}$ under the unique projective map that, for any choice of $\mathsf{e_1}$, maps  $\mathsf{e_1}$ to $\mathsf{c_1}$. 
\end{lemma}

\begin{proof}
The construction, given in the Lemma, for mapping any point $\mathsf{e_1}$ on $\mathsf{E}$ to $\mathsf{c_1}$ on $\mathsf{C}$ defines a one-to-one map $v$ from $\mathsf{E}$ to $\mathsf{C}$.  The tangent point $\mathsf{p}$ is a fixed point of $v$ (consider the limit as $\mathsf{e_1}$ goes to $\mathsf{p}$).  
We claim $v$ is the restriction of a unique projective map, $f$, to the domain $\mathsf{E}$. 
This postulated $f$ obviously has $\mathsf{p}$ as a fixed point, meaning that it is a central projection map, with $\mathsf{p}$ as the centre. 
Say that we can find one such map $f$ from $\mathsf{E}$ to $\mathsf{C}$ and for any point $\mathsf{e_i} \in \mathsf{E}$, the points $\mathsf{p}$, $\mathsf{e_i}$, and $\mathsf{c_i}=f(\mathsf{e_i})\in \mathsf{C}$ are collinear, as in Fig.~\ref{fig-map4}. 
Then $f$ is unique by the fundamental theorem of projective geometry for the plane (see Theorem 3.2.7 of~Ref.~\cite{agoston2005computer}), which says that there is a unique projective transformation that sends four given points ($\mathsf{{a}_1}$, $\mathsf{{b}_1}$, $\mathsf{{c}_1}$, and $\mathsf{{d}_1}$) into another four given points ($\mathsf{{a}_2}$, $\mathsf{{b}_2}$, $\mathsf{{c}_2}$, and $\mathsf{{d}_2}$), if no three in either
set are collinear. (This follows because an ellipse has more than 3 points!) 

Below, we derive an explicit expression for a projective map $f$ with the desired properties, for arbitrary $\mathsf{E}$ and $\mathsf{C}$. 
The projective map we find is a special type of central projection (with centre $\mathsf{p}$) called a planar homology~\cite{hartley03,crannell2019perspective,casse2006projective}. This has the property that, for any point $\mathsf{s}$, the points 
$\mathsf{p}$, $\mathsf{s}$, and $f(\mathsf{s})$ are collinear. Now, $\mathsf{b}$ lies on the line segment $\mathsf{e_1e_2}$ as shown in Fig.~\ref{fig-map4}.  
Since all projective transformations preserve collinearity~\cite{Ric11, agoston2005computer}, $f$ maps the line $\mathsf{e_1e_2}$ to the line $\mathsf{c_1c_2}$. 
Hence it maps $\mathsf{b}$ to a point on $\mathsf{e_1e_2}$. But, since $f$ is a planar homology, this point must also be on the line $\mathsf{pb}$. 
Therefore, this point $f(\mathsf{b})$, must be $\mathsf{h}$, the intersection of $\mathsf{pb}$ and $\mathsf{c_1c_2}$. 
\end{proof}

\begin{figure}[h]\centering
\includegraphics[angle=0,width=0.42\linewidth]{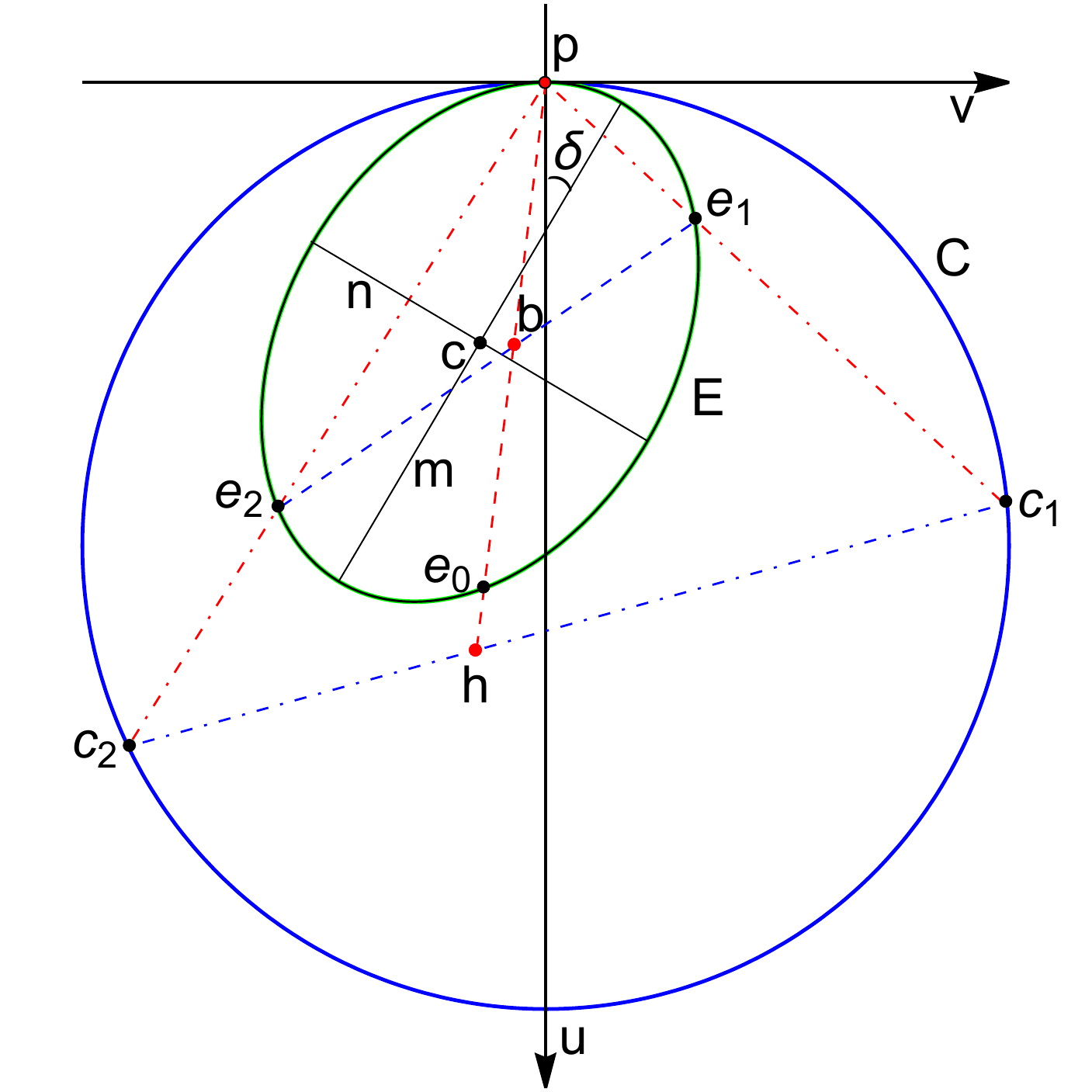}
\vspace{-10pt}
\caption{Planar homology $f$ with centre $\mathsf{p}$ which maps from the ellipse $\mathsf{E}$ to the circle $\mathsf{C}$, so that  $f(\mathsf{e_1})=\mathsf{c_1}$, $f(\mathsf{e_2})=\mathsf{c_2}$, and $f(\mathsf{b})=\mathsf{h}$.}
\label{pbh}
\end{figure}

To find $f$ it will be useful to give explicit quadratic forms for $\mathsf{E}$ and $\mathsf{C}$. As shown in Fig.~\ref{pbh}, consider a circle $\mathsf{C}$ with radius $R$, tangent to the origin $\mathsf{p}$ of the Euclidean plane.
We can express this as the parametric equation in coordinates $u$ and $v$, 
\begin{equation}\label{circle}
(u-R)^2+v^2-R^2=0.
\end{equation}
In terms of homogeneous coordinates, this can equivalently be written as the conic $\vec{\mathsf{x}}^{\top}\!C\vec{\mathsf{x}}=0$, where  $\vec{\mathsf{x}} = (u,v,1)^\top$ and 
\begin{align}
C \coloneqq\begin{pmatrix}
 1 & 0 & -R \\
 0 &  1 & 0 \\
 -R & 0 & 0
\end{pmatrix}.
\label{eq:circle_conic}
\end{align}
For the ellipse $\mathsf{E}$, we suppose its centre is located at $(u_c,v_c)$.
Any such ellipse is described in this plane by
\begin{equation}\label{rotationellipse}
\frac{\left((u-u_c)\cos{\delta}+(v-v_c)\sin{\delta}\right)^2}{m^2} +\frac{\left(-(u-u_c)\sin{\delta}+(v-v_c)\cos{\delta}\right)^2}{n^2}-1=0.
\end{equation}
Here, $m$ and $n$ are the semiaxes of the ellipse and $\delta$ is the angle of clockwise rotation of the major axis with the $u$-axis. 
We are interested in ellipses which are tangent to $\mathsf{C}$ at $\mathsf{p}=(0,0)$. 
By constraining the ellipse to pass through this point, we have \blk
\begin{align}
\left(\frac{{u_c} \cos\delta+{v_c} \sin\delta}{m}\right)^2+\left(\frac{{u_c} \sin \delta-{v_c} \cos \delta}{n}\right)^2=1.
\end{align}
Further, we must have $\dv{u}{v}|_{(0,0)}=0$ so that the ellipse is tangent at this point. 
This means that\blk
\begin{equation}
\frac{1}{2}{u_c}(m^2-n^2)\sin(2\delta)-{v_c}(m^2\cos ^2\delta+n^2\sin ^2\delta)=0.
\end{equation}
Then, by defining 
\begin{align}
G\coloneqq\sqrt{m^2+n^2+(m^2-n^2)\cos(2\delta)}, 
\end{align}
we can write
\begin{align}
u_c=\frac{G}{\sqrt{2}},\quad
v_c=\frac{(m^2-n^2)\sin (2 \delta )}{\sqrt{2}G}.
\end{align}
From this, the ellipse is represented by the conic  $\vec{\mathsf{x}}^{\top}\!E \vec{\mathsf{x}}=0$, with 
\begin{align}
E=\begin{pmatrix}
 1-2\mu & -\nu & -\xi \\
  -\nu &  1 & 0 \\
 -\xi & 0 & 0
\end{pmatrix},
\label{eq:ellipse_conic}
\end{align}
where $\mu\coloneqq(m^2-n^2)\cos (2\delta )/G^2$, $\nu\coloneqq(m^2-n^2)\sin (2\delta )/G^2$, and $\xi\coloneqq2 \sqrt{2} m^2 n^2/G^3$.

Recall that a projective transformation is described by the matrix
\begin{align}
H=\begin{pmatrix}
 h_{11} & h_{12}  & h_{13}  \\
 h_{21}  &  h_{22}  & h_{23}  \\
h_{31}  & h_{32}  & h_{33} 
\end{pmatrix}.
\end{align}
For the conics considered above,  we wish to transform from the ellipse $\mathsf{E}$ to the circle $\mathsf{C}$.
Since this transformation must preserve $\vec{\mathsf{p}} = (0,0,1)^\top$, \emph{i.e.~}$H\vec{\mathsf{p}}=\vec{\mathsf{p}}$, we must have $h_{13}=h_{23}=0$. 
Moreover, the projective map in our construction should transform any point on the ellipse $\mathsf{e_i}$ (other than ${\mathsf{p}}$) to a point $\mathsf{c_i}$ collinear with $\mathsf{p}$ and $\mathsf{e_i}$.
This means that, for any values of $u,v$, $H(u,v,1)^\top \propto (u,v,z)^\top$ for some $z$.
This further restricts the projective transformation to have $h_{12}=h_{13}=h_{21}=h_{23}=0$, and $h_{11}=h_{22}$.
Now, we know that $H$ should transform the matrix $E$ to $C$ via
\begin{equation}
E \mapsto (H^{-1})^\top E H^{-1}=C.
\label{eq:projective_map_matrix}
\end{equation}
Solving this equation for the forms of $E$ and $C$ in Eqs.~\eqref{eq:ellipse_conic} and \eqref{eq:circle_conic}, we find  $h_{11}=h_{22}=1$, $h_{31}=\mu/R$, $h_{32}=\nu/R$ and $h_{33}=\xi/R$. 
Therefore, we can write $H$ as
\begin{align}
H=\begin{pmatrix}\label{pmm}
 1 & 0 & 0 \\
 0 &  1 & 0 \\
\alpha & \beta & \gamma
\end{pmatrix},
\end{align}
where 
\begin{align}\label{abc}
\begin{pmatrix}
 \alpha  \\
 \beta   \\
 \gamma
\end{pmatrix}
=\frac{1}{RG^{2}}\begin{pmatrix}
 (m^2-n^2)\cos (2\delta )  \\
 (m^2-n^2)\sin (2\delta )  \\
2 \sqrt{2} m^2 n^2/G
\end{pmatrix}.
\end{align}
The three eigenvalues of transformation matrix $H$ are 1,1, and $\gamma$. 
Since $H$ has one distinct and two equal eigenvalues, it is a planar homology~\cite{hartley03}.

\subsection{Steering in the Scenario} 
\label{sec:Steering in The Scenario}

From Corollary \ref{corollary} and Lemma \ref{lemma3}, we know that the steerability of a 1PQA in the Scenario is determined entirely by two geometric objects.
These are the point $\mathsf{b}$ (describing Bob's reduced state), and the steering ellipsoid ${\cal E}$, which touches Bob's Bloch sphere at $\mathsf{p}$ (describing the one pure steered state).
Recall that $\cal{E}$ gives information about the centre, orientation and semiaxes of the ellipsoid, which thus also includes the location of $\mathsf{p}$. 
We now present the necessary and sufficient condition for Alice to be able to demonstrate EPR-steering in the Scenario when the plane containing the 1PQA is given. Using Corollary \ref{corollary} and Lemma \ref{lemma3}, we formulate this result as the following Theorem.

\begin{figure}[h]\centering
\includegraphics[angle=0,width=0.44\linewidth]{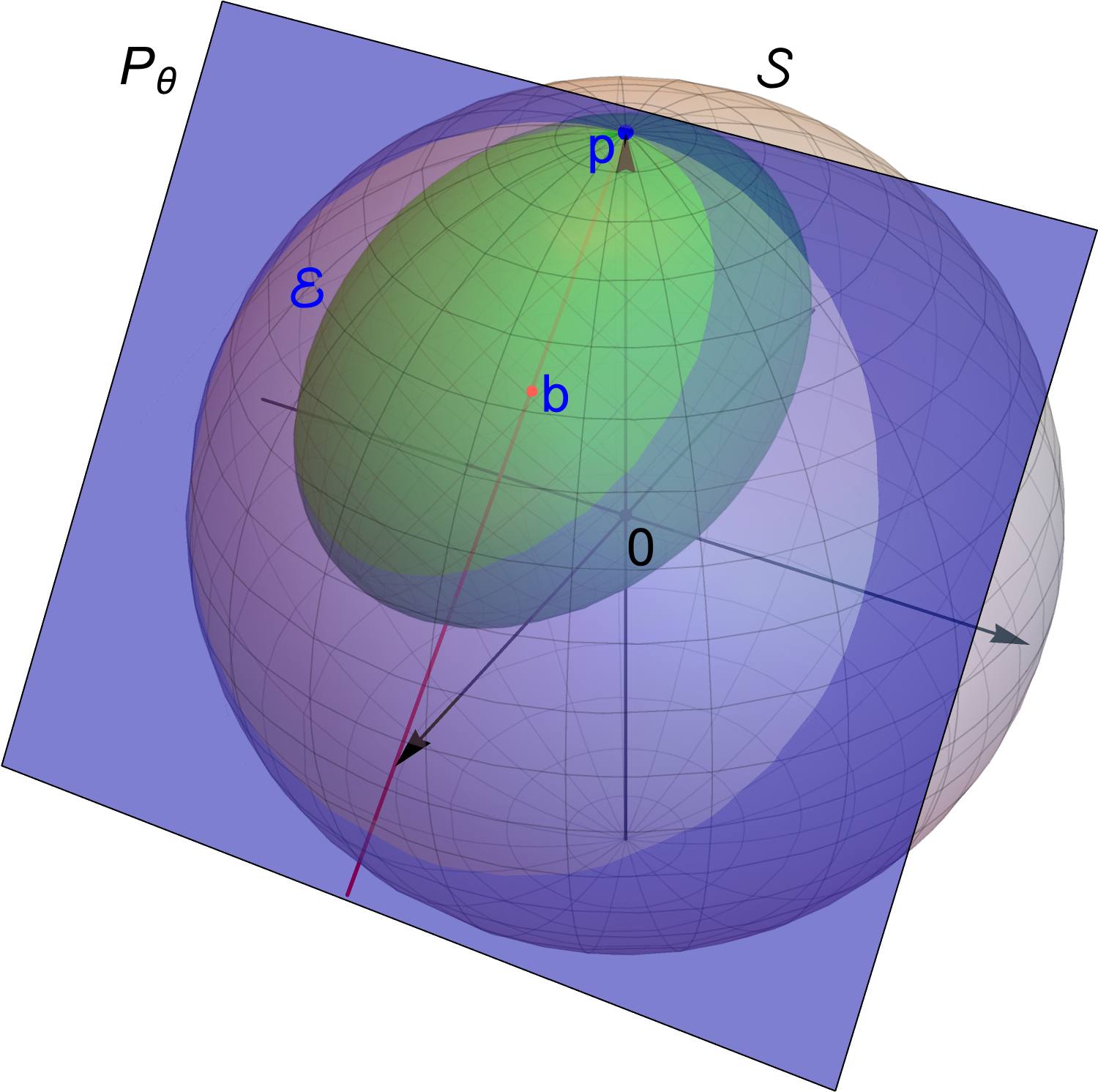}
 \hspace{7pt}
\includegraphics[angle=0,width=0.42\linewidth]{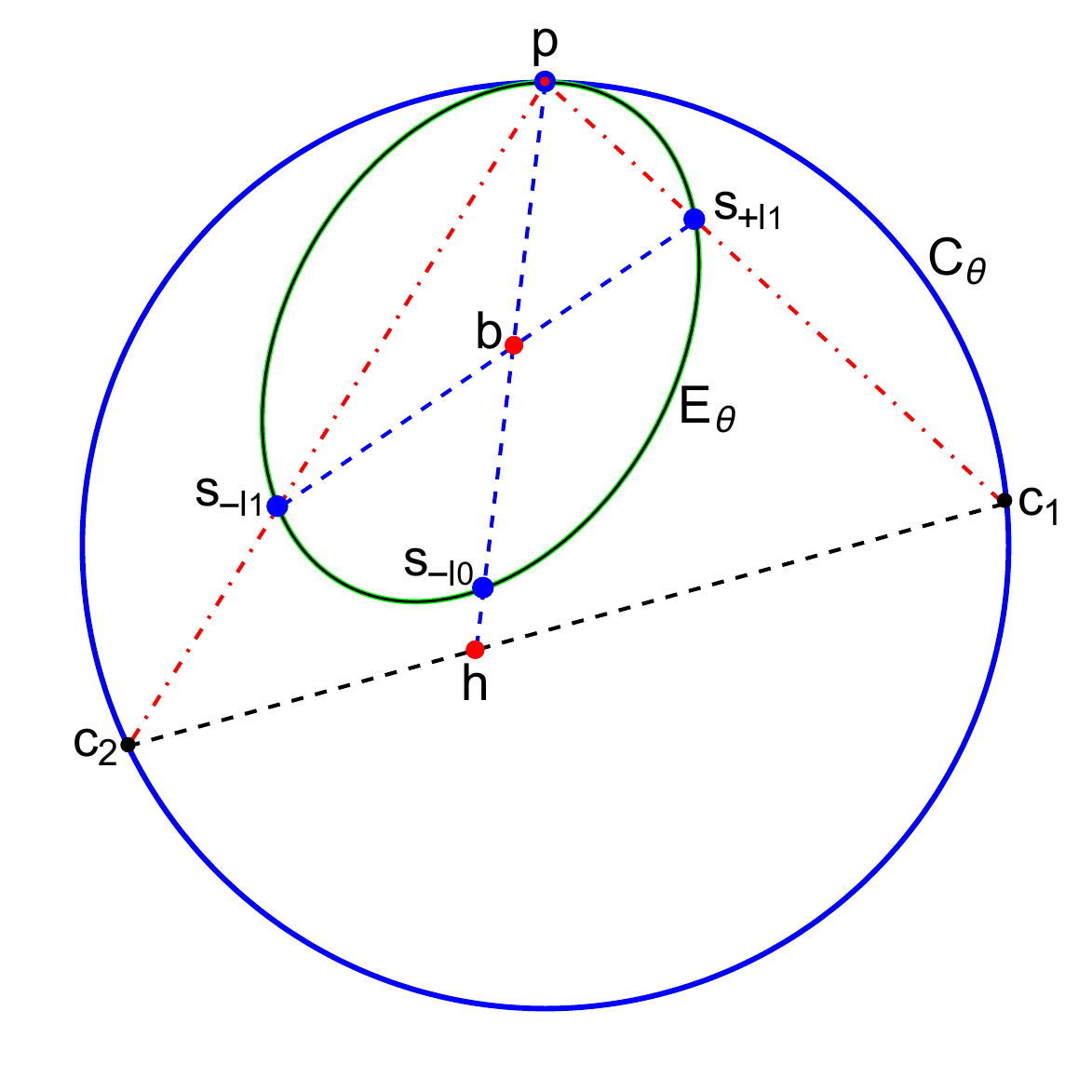}
\vspace{-5pt}
\caption{Illustration of Theorem \ref{theorem}. 
The steering ellipsoid ${\cal E}$ touches the surface of Bob's Bloch sphere ${\cal S}$ at the point $\mathsf{p}$ which represents the pure steered state. 
The red point $\mathsf{b}$ represents the Bloch vector of Bob's reduced state. 
A plane $P_\theta$ contains the line $\mathsf{pb}$, in which the cross sections of ${\cal S}$ and ${\cal E}$ are circles $\mathsf{C_\theta}$ and ellipses $\mathsf{E_\theta}$. 
Four blue points $\mathsf{p, s_{-|0},s_{+|1},s_{-|1}}$ represent four Bob's steered states. The red point $\mathsf{h}$ is the intersection of line $\mathsf{pb}$ and $\mathsf{c_1c_2}$. 
Alice can demonstrate EPR-steering if and only if $\mathsf{h}$ is strictly inside the ellipse $\mathsf{E_\theta}$.
} \label{figgs} 
\end{figure}

\begin{theorem}\label{theorem}
 Consider the one-parameter family of planes ${\cal P}=\{P_\theta| \theta \in [ 0,2\pi)\}$, containing line $\mathsf{pb}$, in which the cross sections of Bob's Bloch sphere ${\cal S}$ and the steering ellipsoid ${\cal E}$ are the circles $\mathsf{C_\theta}$ and ellipses $\mathsf{E_\theta}$. 
Then Alice can demonstrate EPR-steering under the Scenario (Def.~\ref{scenario}) for which the points of the steered states are in the plane $P_\theta$ for some $\theta$, if and only if $f_\theta(\mathsf{b}) \in \mathsf{E_\theta^{\circ}}$, where $\mathsf{E_\theta^{\circ}}$ is the interior of the ellipse $\mathsf{E_\theta}$, 
and $f_\theta$ is the planar homology mapping 
$\mathsf{E_\theta}$ to $\mathsf{C_\theta}$ with centre $\mathsf{p}$.
\end{theorem}
\begin{proof}
Consider all possible steering ellipsoids permitted in the Scenario.
Alice can, by suitably choosing the measurement which \emph{doesn't} steer to the pure state, steer to any point on $\mathsf{E_\theta}$ for any plane $\theta$.
One such plane is illustrated in Fig.~\ref{figgs}, where $\mathsf{b}$ represents Bloch vector of Bob's reduced state, $\mathsf{p}$ is the pure steered state required by the Scenario, and the $\mathsf{s_{-|0},s_{+|1},s_{-|1}}$ represent the three other steered states in Bob's ensembles in this plane.
We know from Lemma \ref{lemma3} that there exists a unique planar homology $f_\theta$ with centre $\mathsf{p}$ mapping $\mathsf{E_\theta}$ to $\mathsf{C_\theta}$, 
and that $\mathsf{h} = f_\theta(\mathsf{b})$. 
Applying Corollary \ref{corollary}, Alice can demonstrate EPR-steering if and only if $\mathsf{h}$ lies between $\mathsf{b}$ and $\mathsf{s_{-|0}}$.
Since $\mathsf{s_{-|0}}$ lies on the boundary of $\mathsf{E_\theta}$. Alice will be able to demonstrate steering iff $f_\theta(\mathsf{b})$ lies strictly inside $\mathsf{E_\theta}$.
\end{proof}

As shown in Figs.~\ref{pbh} and \ref{figgs}, the condition for steerability in Theorem~\ref{theorem} can be equivalently expressed as an inequality involving Euclidean distances for the relevant points, $|{\mathsf{ps_{-|0}}}|>|{\mathsf{ph}}|$.  
That is, $|{\mathsf{pe_{0}}}|>|{\mathsf{ph}}|$ is equivalent to the necessary and sufficient condition $|{\mathsf{ps_{-|0}}}|>|{\mathsf{ph}}|$.  
This allows us to derive a closed-form expression for steerability based on the parametrization of $H$. 
Since the homogeneous coordinate $\vec{\mathsf{b}}=(u_b,v_b,1)^\top$ is mapped to $\vec{\mathsf{h}}=H\vec{\mathsf{b}}$,  using Eq.~\eqref{pmm}  we can get the coordinates of $\mathsf{h} = (u_h, u_v)$ in the $uv$-plane,
\begin{align}
u_h=\frac{u_b}{\alpha u_b+\beta v_b+\gamma},\quad
v_h=\frac{v_b}{\alpha u_b+\beta v_b+\gamma}.
\end{align}
Using the equation of the ellipse, Eq.~\eqref{rotationellipse}, and the equation of the line $\mathsf{pb}$, which is $v=u v_b/u_b$,  the coordinates of their intersection, $\mathsf{e_0}=(u_{e_0}, v_{e_0})$, are 
\begin{align}
u_{e_0}=\frac{2R\gamma u_b^2}{(1-2R\alpha) u_b^2-2R\beta u_b v_b+v_b^2},\quad
v_{e_0}=\frac{2R\gamma u_b v_b}{(1-2R\alpha) u_b^2-2R\beta u_b v_b+v_b^2}.
\end{align}
Since $|{\mathsf{pe_{0}}}|>|{\mathsf{ph}}|$ is equivalent to $u_{e_0}>u_{h}$ in the $uv$-plane, we can reduce this condition to
\begin{align}\label{sinequ}
2 R \left(\gamma ^2+\beta (1+\gamma){v_b}\right){u_b}-\left(1-2 R\alpha  (1+\gamma)\right){u_b}^2-{v_b}^2>0.
\end{align}
This inequality is the necessary and sufficient condition for Alice to be able to demonstrate EPR-steering in the plane $P_\theta$ in the Scenario.

In Theorem \ref{theorem}, different planes $P_\theta$, all containing the line $\mathsf{pb}$, may have different $\mathsf{E}_\theta$ and $\mathsf{C}_\theta$. 
In turn, this will require different homologies $f_\theta$ transforming between the two conics.
To apply Theorem~\ref{theorem}, we can consider the image of $\mathsf{b}$ under all of these maps, and determine steerability.  
We refer to this set, $\{f_\theta(\mathsf{b})\}$, as the \emph{locus} of $\mathsf{h}$, denoted by $\mathsf{L}$ (a line segment).
The position of the locus relative to the steering ellipsoid ${\cal E}$ has three distinct cases relevant to steerability in the Scenario. 
The first of these is where $\mathsf{L}$ is located outside the steering ellipsoid ${\cal E}$.
 In this case, the entire family of planes appearing in the statement of Theorem \ref{theorem} lead to $f_\theta (\mathsf{b})$ being outside the ellipse.
Equivalently, no steering is possible for any choice of second measurement by Alice.  
The second case is where $\mathsf{L}$  penetrates the steering ellipsoid ${\cal E}$. 
In this case, there exists a strict subset of $\theta$ such that a 1PQA in plane $P_\theta$ demonstrates EPR-steering. 
The final case is where $\mathsf{L}$ is contained inside the steering ellipsoid ${\cal E}$. 
Here, any choice of second measurement by Alice will lead to a 1PQA which is steerable. 

 From these observations, we have the following corollary of Theorem \ref{theorem}. 

\begin{corollary}\label{corollary2}
In the Scenario (Def.~\ref{scenario}), consider the one-parameter family of planes ${\cal P}=\{P_\theta| \theta \in [ 0,2\pi)\}$, containing the line $\mathsf{pb}$, in which the cross sections of Bob's Bloch sphere ${\cal S}$ and the steering ellipsoid ${\cal E}$ are circles $\mathsf{C_\theta}$ and ellipses $\mathsf{E_\theta}$.
 Then, 
\begin{enumerate}
\item the condition that, for all $P_{\theta}\in {\cal P}$, $f_\theta(\mathsf{b}) \in \mathcal{E^{\circ}}$ is {\em sufficient} for Alice to demonstrate EPR-steering of Bob, and
\item the existence of a plane $P_\theta\in {\cal P}$ such that $f_\theta(\mathsf{b}) \in \mathsf{E_\theta^{\circ}}$, is {\em necessary} for Alice to demonstrate EPR-steering of Bob. 
\end{enumerate}
where $\mathcal{E^{\circ}}$ (resp.  $\mathsf{E_\theta^{\circ}}$) represents the interior  
of the ellipsoid $\mathcal{E}$ (ellipse $\mathsf{E_\theta}$), 
and $f_\theta$ is the planar homology mapping  $\mathsf{E_\theta}$ to $\mathsf{C_\theta}$ with centre $\mathsf{p}$.
\end{corollary}

\section{Applications}             
 \label{sec:examples}
 
We now apply Theorem \ref{theorem}  to three different cases of two-qubit states.
These are tangent X-states, canonical obese states, and tangent sphere states.

\subsection{Tangent X-states}
\label{sec:tangent_x_states}
We first consider the subset of two-qubit X-states that permit Alice to steer to exactly one pure state. 
For brevity, we call such states tangent X-states. 
X-states are the states which the density matrix contains only non-zero elements along the main diagonal and anti-diagonal~\cite{Yu07}, 
\begin{align}
\rho_\text{X}=\begin{pmatrix}
\rho_{11} & 0 & 0 & \rho_{14} \\
0 & \rho_{22} & \rho_{23} & 0 \\
0 & \rho_{23}^\ast & \rho_{33} & 0 \\
\rho_{14}^\ast & 0 & 0 & \rho_{44}
\end{pmatrix},
\end{align}
which has seven real parameters, counting normalization. 
It is worth noting that maximally entangled Bell states~\cite{Nie02}
and ‘Werner’ states~\cite{Wer89} are particular cases of X-states. 
An algebraic characterisation of X-states in quantum information is presented in Ref.~\cite{Rau09}.
Under appropriate local unitary transformations, which preserve EPR-steering~\cite{Qui15}, all elements can be transformed into real numbers. 
Hence we only need to consider the following density matrix with five real parameters by setting $\boldsymbol{a}=(0,0,a)^\top$, $\boldsymbol{b}=(0,0,b)^\top$ and ${T}=\operatorname{diag}\left(t_{x}, t_{y}, t_{z}\right)$. 
This means the X-states we consider have the form
\begin{align}
\rho_\text{X}=\frac{1}{4}\left({I} \otimes {I}+a \sigma_{z} \otimes {I}+{I} \otimes b \sigma_{z}+\sum_{i=x,y,z} t_{i} \sigma_{i} \otimes \sigma_{i}\right),\label{x-states}
\end{align}
According to Eq.~\eqref{onepure1}, an X-state of this form will have a pure steered state if 
\begin{align}\label{xsone}
t_z=1+a-b,
\end{align}
From Eqs.~\eqref{xsone} and \eqref{orientationmatrix}, the three semiaxes of the quantum steering ellipsoid are 
\begin{align}
n_x=\frac{|t_x|}{\sqrt{1-a^2}},\quad
n_y=\frac{|t_y|}{\sqrt{1-a^2}},\quad
m=\frac{1-b}{1-a}.
\end{align}
Furthermore, the semiaxes of $\mathsf{E}_\theta$ formed by such X-states and ${P_\theta}$ are given by $m$ and
\begin{align}\label{ellipsemn}
n_\theta=\frac{\tau_\theta}{\sqrt{1-a^2}},\quad
\end{align}
where
\begin{align}
\tau_\theta=\frac{|t_xt_y|}{\sqrt{t_x^2(\sin{\theta})^2+t_y^2(\cos{\theta})^2}},
\end{align}
and one of Alice's measurements consists of effects which are eigenstates of $\sigma_z$. 
 Define $\theta$ as the angle of the plane (appearing in Theorem \ref{theorem}) from the $x$-axis, restricted to the interval $[0,\pi/2]$ without loss of generality. In addition, we assume without loss of generality that $|t_x|>|t_y|$, so that $\max_\theta n_\theta=n_x$ and $\min_\theta n_\theta=n_y$.

Since these states satisfy the requirements of the Scenario, we can apply Theorem \ref{theorem}. 
First, note that the intersection of any $P_\theta$ with the Bloch sphere gives the same circle, and so we drop the dependence of $\mathsf{C}$ on $\theta$.
 For tangent X-states, we have $R=1$, $\delta=0$ and $G=\sqrt{2}m$, and so the parameters in \eqref{abc} defining the required projective map $f_\theta$ are 
\begin{align}\label{xabc}
\begin{pmatrix}
 \alpha  \\
 \beta  \\
 \gamma
\end{pmatrix}
=\frac{1}{2m^{2}}\begin{pmatrix}
 m^2-n_\theta^2\\
 0  \\
2 m n_\theta^2
\end{pmatrix}.
\end{align}
We know that $\mathsf{b}$ can be expressed in Cartesian coordinates as $(u_b,v_b) = (1-b,0)$.
Using Eq.~\eqref{sinequ}, we immediately see that Alice can demonstrate EPR-steering in plane $P_\theta$ under the Scenario if and only if
\begin{align}\label{xs-steering}
\tau_\theta^2>(1-b)(b-a).
\end{align}
We can gain further insight about the importance of the plane which contains the mixed ensemble by studying the locus $\mathsf{L}$ of $\mathsf{h}$.
This is the set of points
\begin{equation}
    \mathsf{L} = \left\{ \left(\frac{2(1-b)^2}{\tau_\theta^2+(1-b)^2},0\right)\right\}_\theta.
\end{equation}
For tangent X-states, all points belonging to the locus lie on the $u$-axis.
In this case, there is a critical $\theta\in[0,\pi/2]$, corresponding to the point in the locus at $\mathsf{s}_{-|0}$, where equality holds between the left- and right-hand sides of Eq.~\eqref{xs-steering}. 
We can easily find that the locus contains all points on the line segment between points $\left(\frac{2(1-b)^2}{t_x^2+(1-b)^2},0\right)$ and $\left(\frac{2(1-b)^2}{t_y^2+(1-b)^2},0\right)$. 
If $t_x=t_y$, the locus reduces to a single point. 

Another derivation of \eqref{xs-steering} that uses only elementary geometry is as follows. 
Consider the unit circle $\mathsf{C}$ and an ellipse $\mathsf{E_\theta}$ with two semiaxes $m$ and $n_\theta$ shown in Fig.~\ref{X-states-geo}.
Note that the $m$-semiaxis aligns with the $z$-axis for tangent X-states.
The tangent point where $\mathsf{E_\theta}$ is tangent to $\mathsf{C}$, $\mathsf{p}$, represents the pure state which can be steered to.
\begin{figure}[h]\centering
\includegraphics[angle=0,width=0.42\linewidth]{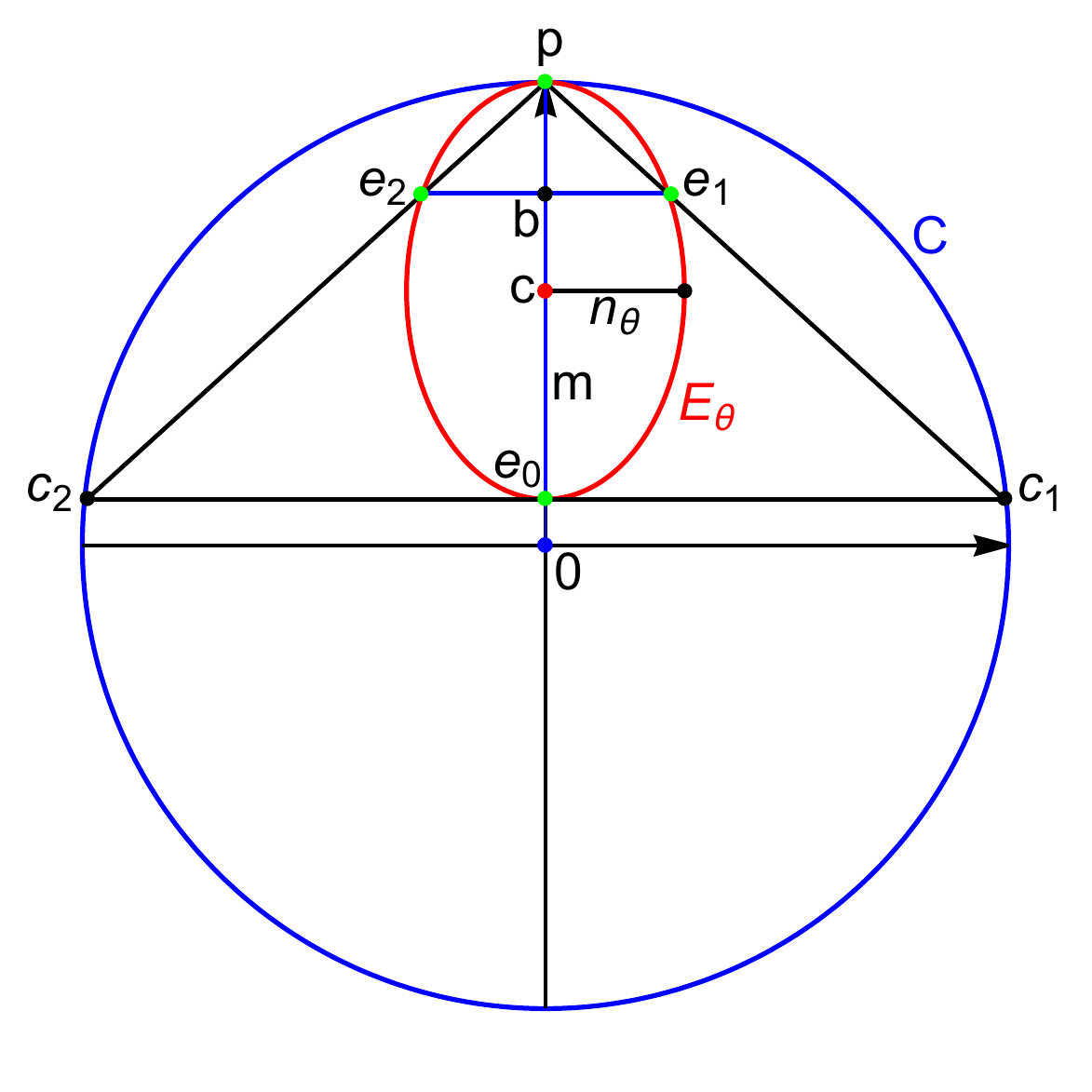}
\vspace{-10pt}
\caption{Geometric representation of steering with tangent X-states. 
Here, $m$ and $n_\theta$ are two semiaxes of the steered ellipse $\mathsf{E_\theta}$.
} 
\label{X-states-geo}
\end{figure}
According to Corollary \ref{corollary} and Lemma \ref{lemma3}, for a given plane ${P}_\theta$, the change of the position of the mixed ensemble (the one without the pure state) \emph{within} this plane does not change steerability. 
Therefore, we can choose the position of the second ensemble to be parallel to the $xy$-plane of the Bloch ball. 
From Fig.~\ref{X-states-geo}, we can  easily find the maximal distance between $\mathsf{p}$ and $\mathsf{b}$ for which $\mathsf{e}_0$ is outside the convex hull of the triangle $\triangle\mathsf{pc_1 c_2}$.
By Corollary \ref{corollary}, we know that for any $\mathsf{b}$ beyond this distance, the 1PQA will be non-steerable. 
This critical distance ${b}_c$ is given by
\begin{align}
{b}_c=1-\frac{2mn_\theta^2}{n_\theta^2+m-m^2}.
\end{align}
Hence Alice can demonstrate EPR-steering if and only if
\begin{align}\label{xs-steering2}
1-b<\frac{2mn_\theta^2}{n_\theta^2+m-m^2},
\end{align}
From Eq.~\eqref{ellipsemn} we can easily see that condition \eqref{xs-steering2} is equivalent to \eqref{xs-steering}.

\subsection{Canonical maximally obese states}
\label{sec:can_obese}
An interesting class of steering ellipsoids are those with the largest volume given its centre $\boldsymbol{c} = (0, 0, c)$, under the guarantee that the steering ellipsoid is physical.
These ellipsoids, $\mathcal{E}_{c}^{\max}$, are called \emph{maximally obese}~\cite{Mil14, Mil14a}, and correspond to the single-parameter $(0\leq c\leq 1)$ two-qubit state called canonical maximally obese states
\begin{equation}\label{cmos}
    \rho_c^\mathrm{max} (c) \coloneqq \left(1-\frac{c}{2} \right)\ketbra{\psi_c} + \frac{c}{2}\ketbra{00},
\end{equation}
where $\ket{\psi_c}\coloneqq (2-c)^{-1/2}(\sqrt{1-c} \ket{01} +\ket{10})$.
Maximally obese states have been previously studied in the literature, and were found to maximize concurrence, in addition to other measures of quantum correlations\cite{Mil14a}.  

This state has a steering ellipsoid with two equal major semiaxes $n_{1}=n_{2}=\sqrt{1-c}$ and minor semiaxis $m=1-c$. 
This is a special case of the tangent X-states considered above, allowing us to use  Eq.~\eqref{xs-steering2} to obtain 
\begin{align}\label{lvpe}
b>\frac{3c-1}{1+c}. 
\end{align}
From Eq.~\eqref{cmos}, we obtain $b=c$,
 and so Eq.~\eqref{lvpe} becomes
$c>\ro{3c-1}/\ro{1+c}$,
which is equivalent to
\begin{align}\label{lvpe3}
c<1,
\end{align}
This is always satisfied by maximally obese states when the steering ellipsoid has nonzero volume, $c\neq1$.  
Hence, the canonical maximally obese state is always steerable by a two-setting strategy, no matter how small the steering ellipsoid is, provided it does not reduce to a point (when the state reduces to the product state $\ket{00}$). 
 It is known (see Theorem~2 of Ref.~\cite{Mil14a}) that violation of the CHSH inequality can only be achieved for these states
when the centre of the ellipsoid satisfies $0\leq c\leq 1/2$. \blk
Steering, on the other hand, can be achieved for any $0\leq c < 1$, by steering to a pure state and \emph{any} other measurement.

\subsection{Tangent Sphere States}
\label{sec:tangent_spheres}

If the steering ellipsoid is a sphere with radius $r$, due to symmetry, all the planes passing through the tangent point $\mathsf{p}$ are equivalent to the single-parameter set
\begin{align}
\mathbb{P}=\{P_\theta \mid \theta\in[-\pi,\pi]\}.
\end{align}
where $\theta$ is the angle between the plane's normal vector and the $z$-axis.
We assume $\mathsf{b}$ can be located anywhere inside this $\mathcal{E}$, and refer to these ellipsoids as those corresponding to \emph{tangent sphere states}.
Considering Theorem \ref{theorem}, we can reverse the construction which utilizes projective transformations to certify steerability in the Scenario, as follows.
To apply the theorem with knowledge of $\mathsf{b}$, one needs to find its location after applying the projective map $f_\theta$.
However, knowing the geometry of $\mathsf{E}_\theta$ and $\mathsf{C}_\theta$ is sufficient to derive $f_\theta$.
By Theorem \ref{theorem}, if $P_\theta$ and $\mathsf{E}_\theta$ permit a steerable 1PQA, then $f_\theta(\mathsf{b})$ must be interior to $\mathsf{E}_\theta$.
Since this transformation is invertible, we can find the set of $\mathsf{b}$, which remain interior to $\mathsf{E}_\theta$ under $f_\theta$, by applying $f^{-1}_\theta$ to $\mathsf{E}_\theta$.

\begin{figure}[h]
\centering
\includegraphics[angle=0,width=0.33\linewidth]{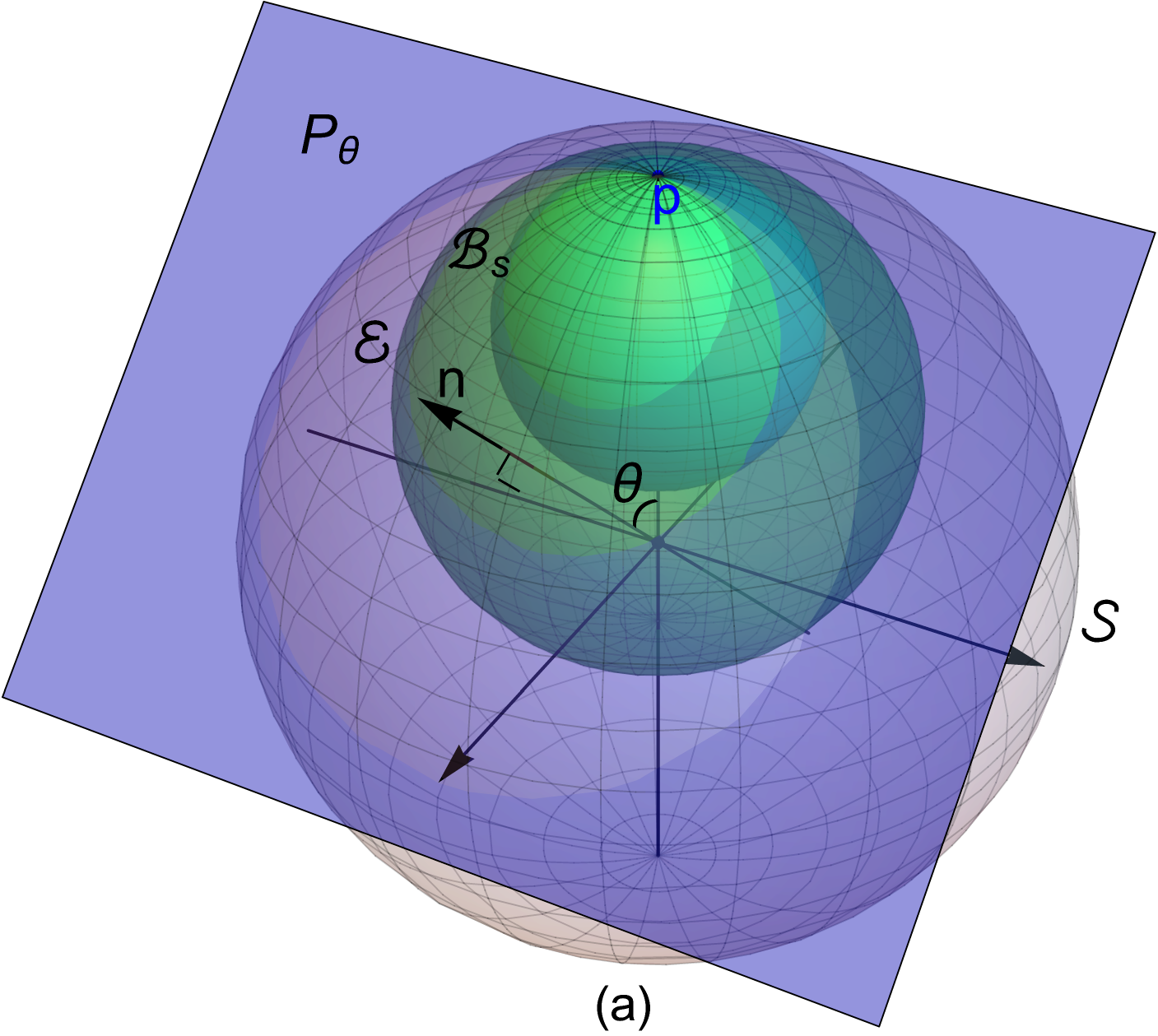}
\hspace{2pt}
\includegraphics[angle=0,width=0.27\linewidth]{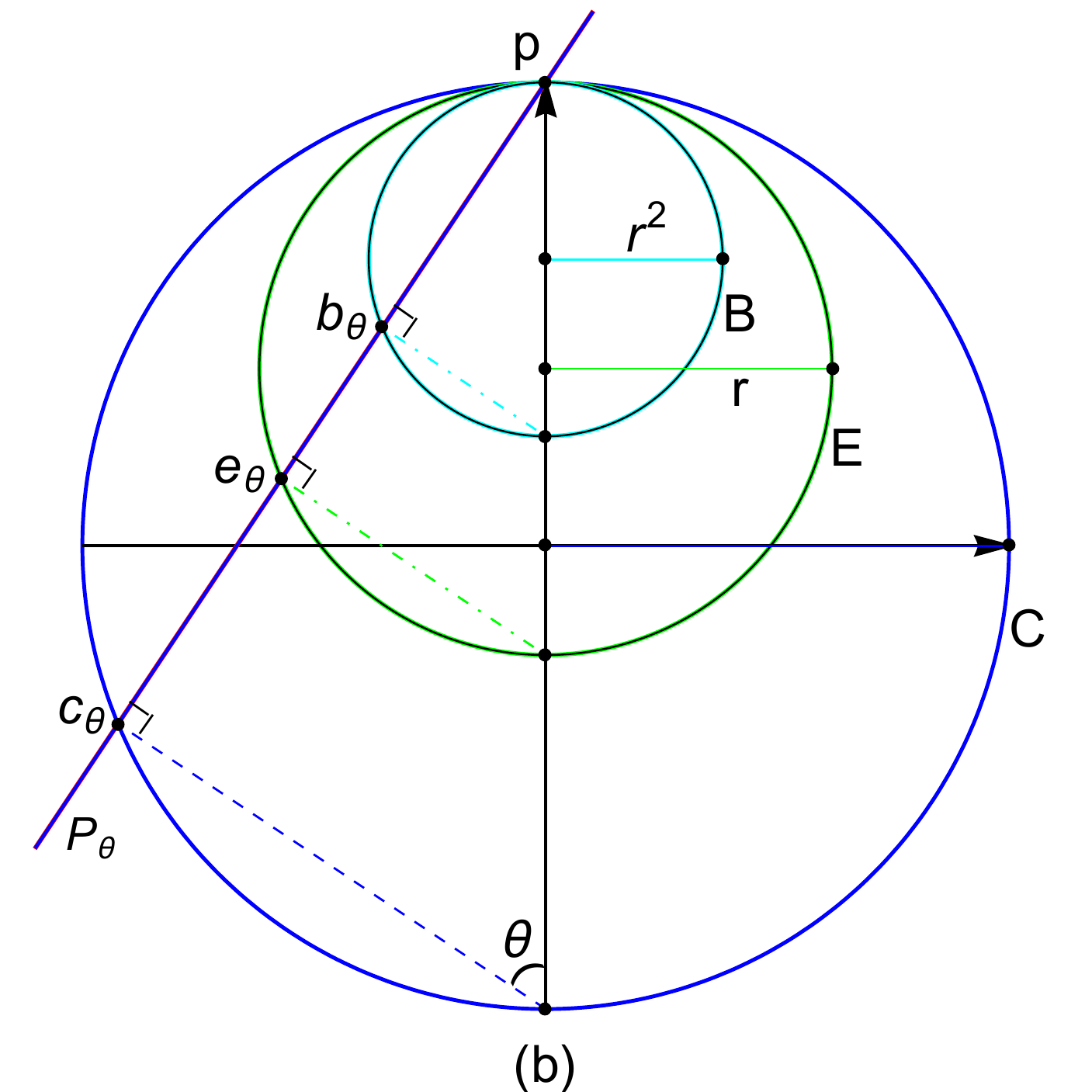}
\hspace{2pt}
\includegraphics[angle=0,width=0.27\linewidth]{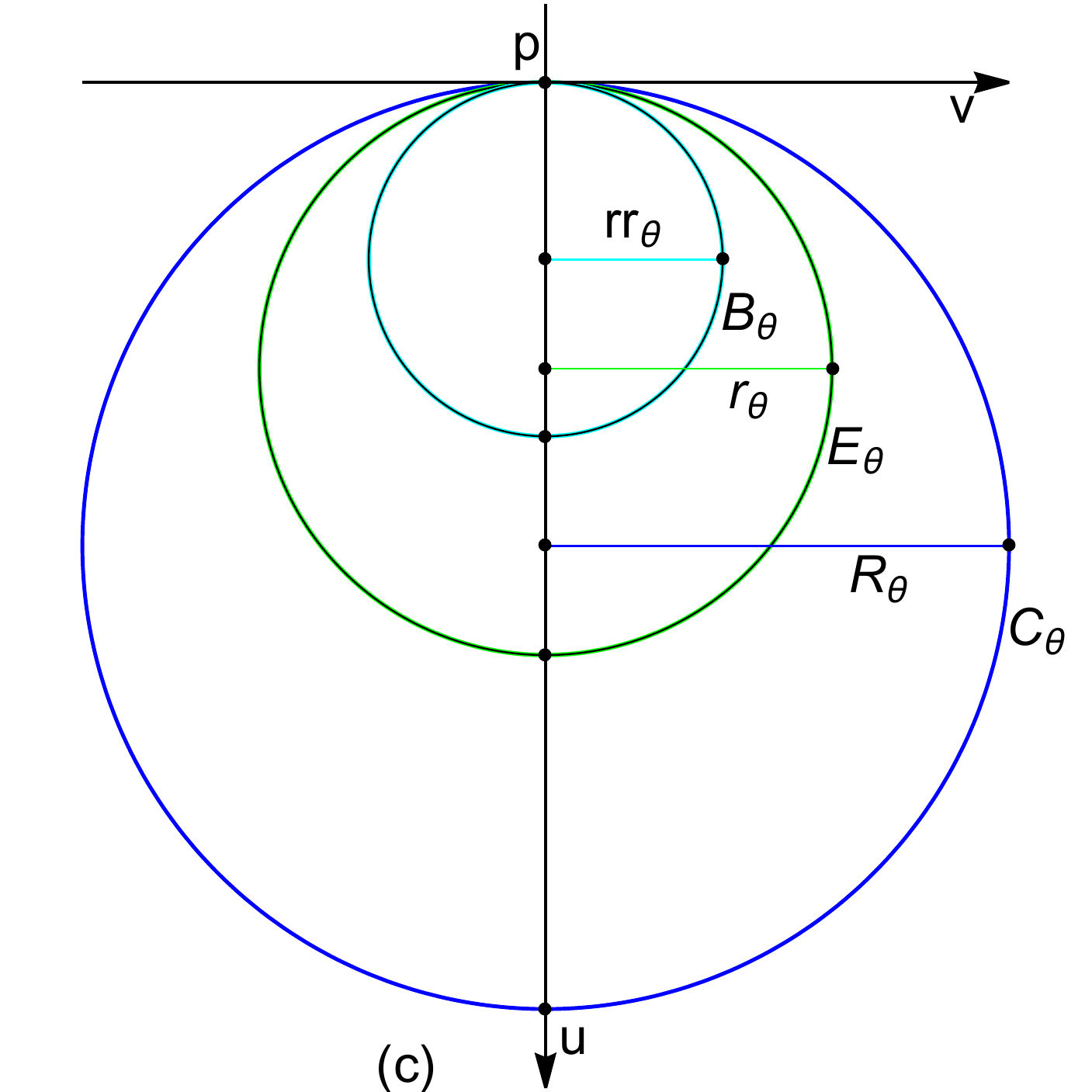}
\caption{
 Steering geometry for tangent sphere states which we use to prove steerability, using projective map inversions. 
If Bob's reduced state is inside to the inner sphere ${\cal B}_\mathrm{s}$, Alice can always steer with a two-setting strategy, as per the Scenario.
(a) A tangent-sphere steering ellipsoid $\mathcal{E}$ which touches the Bloch sphere $\mathcal{S}$ at $\mathsf{p}$.
The plane $P_\theta$ intersects these objects at $\mathsf{p}$, and is specified by $\theta$, the angle between a vector normal to the plane, and the $z$-axis.
(b) Cross section of the Bloch ball showing the position of $P_\theta$ (purple line) in relation to the Bloch sphere and the steering ellipsoid.
(c) Cross sections in the $P_\theta$ plane.
As usual, we work in the Cartesian coordinate system with the origin at the point $\mathsf{p}$.  
The cross sections of the plane $P_\theta$ and the spheres ${\cal S}$, ${\cal E}$ and ${\cal B}_\mathrm{s}$ are circles $\mathsf{C}_\theta$,  $\mathsf{E}_\theta$, and  $\mathsf{B}_\theta$, respectively.}
\label{spheres3}
\end{figure}

Consider Fig.~\ref{spheres3}(a), which illustrates an example of a steering sphere tangent to the Bloch sphere.
For any ${P}_\theta$ which cuts at the tangent point, the cross section of $\mathcal{S}$ and $\mathcal{E}$ will be circles $\mathsf{C}_\theta$ and  $\mathsf{E}_\theta$, with respective radii $R_\theta$ and $r_\theta$, with $R_\theta>r_\theta$. 
 In this case, the required planar homology, which maps $\mathsf{E}_\theta$ to $\mathsf{C}_\theta$,  in Eq.~\eqref{pmm}  is specified by $m=n=r_\theta$, $\delta=0$ and $G=\sqrt{2}r_\theta$, and so 
\begin{align}
H_{\theta}=\text{diag}\left(1,1,\frac{r_\theta}{R_\theta}\right).
\end{align}
The inverse of this matrix is
\begin{align}
H_{\theta}^{-1}=\text{diag}\left(1,1,\frac{R_\theta}{r_\theta}\right).
\end{align}
By reversing the projective map construction, and applying the map which takes $\mathsf{C}_\theta\mapsto\mathsf{E}_\theta$, we can find another conic nested inside it, which contains the set of $\mathsf{b}$ for which the 1PQA is steerable.
To this end, we observe Fig.~\ref{spheres3} (c) shows the intersection of $P_\theta$ with $\mathcal{S}$ and $\mathcal{E}$. 
The radius of $\mathsf{C}_\theta$ in this plane is given by $R_\theta=\sin\theta$, and the radius of  $\mathsf{E}_\theta$ is $r_\theta=r\sin\theta$. 
Therefore, the ellipse corresponds to the symmetric matrix
\begin{align}
E_\theta=\begin{pmatrix}
 1 & 0 & -r_\theta \\
  0 &  1 & 0 \\
 -r_\theta & 0 & 0
\end{pmatrix}.
\end{align}
By inverting the projective map and transforming this conic, we find
\begin{align}
H_\theta^\top E_\theta H_\theta =\begin{pmatrix}
 1 & 0 & -r^2\sin\theta \\
  0 &  1 & 0 \\
 -r^2\sin\theta & 0 & 0
\end{pmatrix},
\end{align}
which is a circle of radius $r^2\sin\theta$ in the plane $P_\theta$. 
By varying the plane $P_\theta$ over $\theta$, the union of all such inner circles forms a sphere of radius $r^2$ tangent to both the Bloch sphere and steering ellipsoid at $\mathsf{p}$, as shown in Fig.~\ref{spheres3}(a).
By Theorem~\ref{theorem}, any $\mathsf{b}$ inside this inner sphere will give rise to 1PQAs by which Alice steers Bob in the Scenario.

We can further relate ${\cal B}_\mathrm{s}$ to the probability that the steered state at $\mathsf{p}$ appears in one of Bob's ensembles.
Define $\mathsf{q}$ as a point of intersection between the line passing through $\mathsf{pb}$ and the surface of the steering ellipsoid ${\cal E}$.
If $\mathsf{b}$ is interior to the inner sphere ${\cal B}_\mathrm{s}$, we have
\begin{align}\label{ball1}
\frac{|{\mathsf{pb}}|}{|{\mathsf{pq}}|}<r.
\end{align}
We also know that 
\begin{align}\label{ball2}
{|{\mathsf{pb}}|}=(1-p_\mathsf{p}){|{\mathsf{pq}}|},
\end{align}
where $p_\mathsf{p}$ is the probability of Alice steering Bob’s system to the pure state. 
Hence, from the above two Eqs.~\eqref{ball1} and \eqref{ball2}  we obtain
\begin{align}\label{tssp}
p_\mathsf{p} > 1-r
\end{align}
as the condition required for steering in the Scenario for tangent sphere states.
In words, this inequality means that if and only if the probability $p_\mathsf{p}$ of Alice steering Bob’s system to a pure state is greater than $1-r$, with $r$ the radius of his steering sphere, then Alice can demonstrate EPR-steering in the Scenario (Def.~\ref{scenario}). 

\section{The power of one pure steered state}
\label{sec:one_pure_state_power}

In the previous section, we inverted the construction based on projective maps to find a range of probabilities for Bob's pure steered state to be necessary and sufficient to prove steerability in the Scenario, for the particular case of tangent sphere states. 
We finish by generalizing this result to arbitrary states with tangent steering ellipsoids.
The main result of this section is the following theorem, which contains one necessary and one sufficient condition, based on the location of Bob's reduced state.

\begin{theorem}\label{theorem2}
Consider the Scenario (Def.~\ref{scenario}), wherein $p_\mathsf{p}$ denotes the probability that Bob is steered to a pure state for one of Alice's measurements.  Given only the steering ellipsoid, ${\cal E}$, 
\begin{enumerate}
\item there exists a $p_{\text{max}}^{\mathcal{E}}\in [0, 1)$ such that $p_\mathsf{p}>p_{\text{max}}^{\mathcal{E}}$ is a {\em sufficient} condition for Alice to demonstrate EPR-steering using only this measurement and {\em any} one other dichotomic projective measurement;
\item there exists a $p_{\text{min}}^{\mathcal{E}}\in [0, p_{\text{max}}^{\mathcal{E}}]$ such that $p_\mathsf{p}>p_{\text{min}}^{\mathcal{E}}$ is a {\em necessary} condition for Alice to demonstrate EPR-steering using only this measurement and {\em some} other dichotomic projective measurement.
\end{enumerate}
\end{theorem}

\begin{figure}[h]\centering
\includegraphics[angle=0,width=0.42\linewidth]{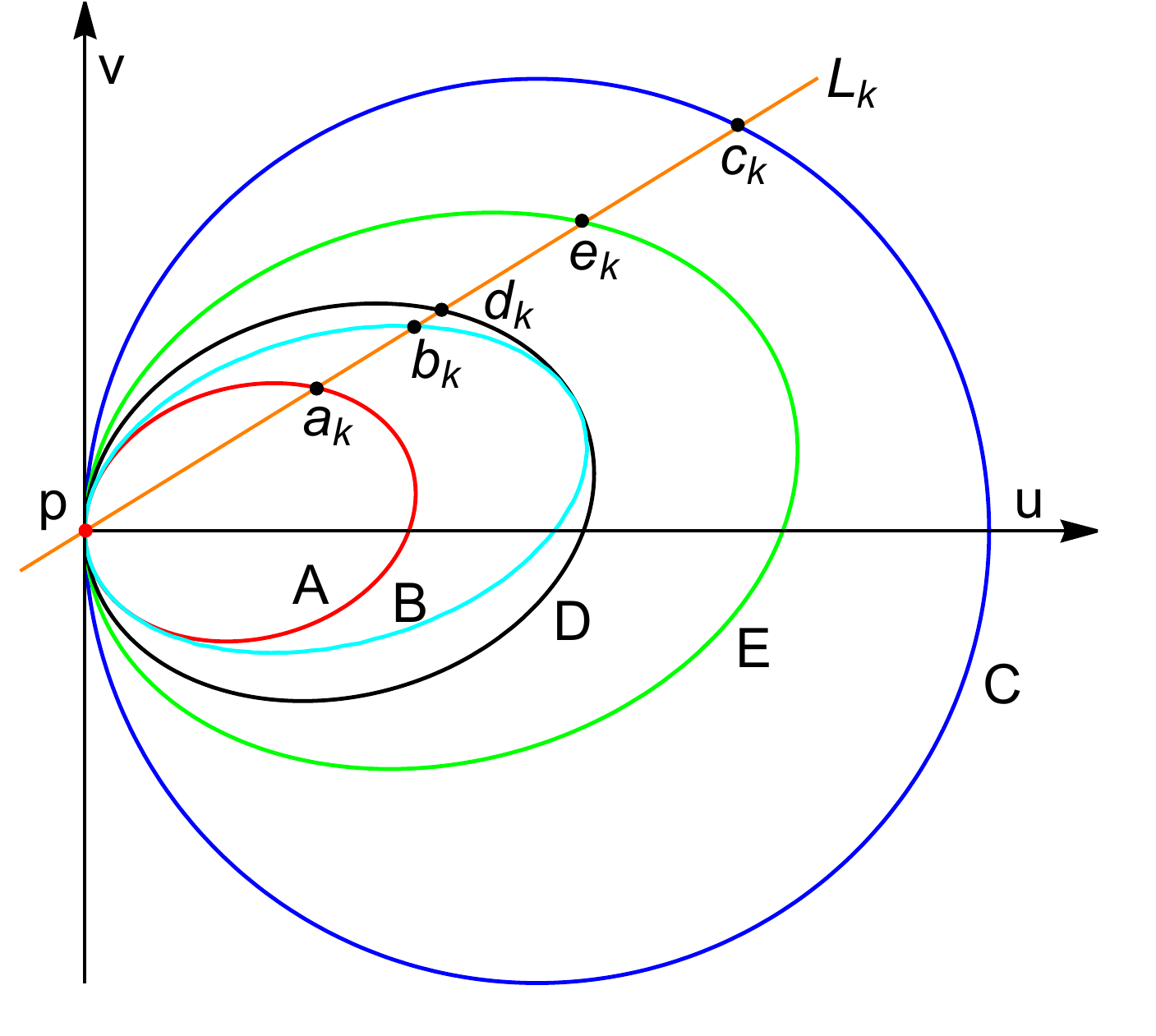}
\vspace{-10pt}
\caption{
The cross sections of a general plane $P_{\theta\phi}$ with the sphere ${\cal S}$ and steering ellipsoid ${\cal E}$.
The cross sections of these are denoted by the circle $\mathsf{C}$ and ellipse $\mathsf{E}$, respectively. 
The ellipse $\mathsf{B}$ comes from applying to $\mathsf{E}$ the inverse of the planar homology $f$ of centre $\mathsf{p}$ that maps $\mathsf{E}$ to $\mathsf{C}$.  
The ellipses $\mathsf{A}$ and $\mathsf{D}$ are obtained by uniformly shrinking $\mathsf{E}$ towards $\mathsf{p}$ to find nested ellipses which are contained within it, and contain it, respectively. 
The ellipses A and D relate directly to $p^\mathcal{E}_{\text{max}}$ and $p^\mathcal{E}_{\text{min}}$ appearing in Theorem~\ref{theorem2}, respectively. 
}
\label{five}
\end{figure}

\begin{proof}
Consider the two-parameter family of planes $P_{\theta\phi}$ whose intersection with the Bloch sphere contains the tangent point $\mathsf{p}$. 
 Each of these can be defined by the normal vector $\hat{\boldsymbol{n}}_{\theta\phi}=(\sin\theta\cos\phi,\sin\theta\sin\phi,\cos\theta)$, $\theta\in [0,\pi], \phi\in [0,2\pi)$.
 For the Scenario (Def.~\ref{scenario}), one such plane is shown in Fig.~\ref{five}, where the cross sections of Bob’s Bloch sphere $\mathcal{S}$ and the steering ellipsoid $\mathcal{E}$ are circles $\mathsf{C_{\theta\phi}}$ and ellipses $\mathsf{E_{\theta\phi}}$. 
For convenience, we will omit the parameters $\theta$ and $\phi$ in the notation when there is no risk of confusion. 
Using the matrix form of $H$ (see Eq.~\eqref{pmm}),  we can invert the projective map and apply it to $\mathsf{E}$, yielding a new ellipse $\mathsf{B}$, defined by  $B=H^{\top}EH$. 
We know from Theorem \ref{theorem} that Alice can demonstrate EPR-steering in the Scenario if and only if the point $\mathsf{b}$, corresponding to  Bob's reduced state, is interior to $\mathsf{B}$. 

In each plane, we consider a $u$-$v$ coordinate system defined such that its origin is at $\mathsf{p}$ and the centre of $\mathsf{C}$ lies along the positive half of the $u$-axis. 
Consider all lines $\mathsf{L}_k$ ($v=ku$) with finite slope $k$ passing through $\mathsf{p}$. Besides $\mathsf{p}$, $\mathsf{L}_k$ intersects the circle $\mathsf{C}$ at the point $\mathsf{c}_k$, and the ellipses $\mathsf{E}$ and $\mathsf{B}$ at points $\mathsf{e}_k$ and $\mathsf{b}_k$, respectively. 
In the $uv$-plane, 
this point of intersection with $\mathsf{C}$ is given by $\mathsf{c}_k=(u_{\mathsf{c}_k},v_{\mathsf{c}_k})$, where
\begin{align}
u_{\mathsf{c}_k}=\frac{2R}{1+k^2},\quad v_{\mathsf{c}_k}=\frac{2kR}{1+k^2}.
\end{align}
Now, the inverse of the projective transformation matrix in Eq.~\eqref{pmm} is
\begin{align}
{H}^{-1}=\begin{pmatrix}
 1 & 0 & 0 \\
 0 &  1 & 0 \\
-\frac{\alpha}{\gamma} &-\frac{\beta }{\gamma}  & \frac{1}{\gamma}
\end{pmatrix}.
\end{align}
 Therefore, we can get the coordinates of the point where $\mathsf{L}_k$ intersects $\mathsf{E}$ (other than $\mathsf{p}$) as $\vec{\mathsf{e}}_{k}={H}^{-1}\vec{\mathsf{c}}_{k}$,
where $\vec{\mathsf{c}}_{k}=(u_{\mathsf{c}_k},v_{\mathsf{c}_k},1)^\top$, and
\begin{align}
&u_{\mathsf{e}_k}=\frac{2R \gamma}{1+k^2-2R(\alpha+k\beta)},\quad
v_{\mathsf{e}_k}=\frac{2kR \gamma}{1+k^2-2R(\alpha+k\beta)}.
\end{align}
We can also obtain the coordinates of the point $\vec{\mathsf{b}}_{k}={H}^{-1}\vec{\mathsf{e}}_{k}$, i.e.,
\begin{align}
u_{\mathsf{b}_k}=\frac{2R \gamma^2}{1+k^2-2R(1+\gamma)(\alpha+k\beta)},\quad
v_{\mathsf{b}_k}=\frac{2kR \gamma^2}{1+k^2-2R(1+\gamma)(\alpha+k\beta)}.
\end{align}

We know that a 1PQA prepared for Bob in the context of the Scenario is steerable iff $\mathsf{b} \in \mathsf{B^{\circ}} \coloneqq f^{-1}(\mathsf{E^{\circ}})$; c.~f.~Theorem~\ref{theorem}.
Any valid $\mathsf{b}$ must lie along one of the lines $\mathsf{L}_k$.
For any such line, we observe that the position of $\mathsf{b}$, relative to $\mathsf{p}$ and $\mathsf{e}_k$, determines the probability of Bob's pure state, $p_\mathsf{p}\coloneqq p(k, \mathsf{b})$, appearing in that ensemble.
Specifically, this probability is
\begin{equation}\label{pp}
p(k, \mathsf{b}) = 1-\frac{|\mathsf{pb}|}{|\mathsf{pe}_k|},
\end{equation}
where we have explicitly included $\mathsf{b}$ as an argument in the function, for clarity below.
Theorem~\ref{theorem2} implies that $\mathsf{b}\in\mathsf{L}_k$ corresponds to a steerable tangent steering ellipsoid in the Scenario iff it is on the open line segment $\mathsf{pb}_k$. 
This is equivalent to the condition $|\mathsf{pb}|<|\mathsf{pb}_k|$, or, expressed as an inequality involving the probability of Bob's pure steered state,
\begin{equation}
p(k, \mathsf{b}) > p(k, \mathsf{b}_k),
\label{eq:suff_prob_condition}
\end{equation}
where
\begin{align}
p(k, \mathsf{b}_k)
&=1-\frac{|\mathsf{pb}_k|}{|\mathsf{pe}_k|} \\
&=\frac{(1+k^2)(1-\gamma)-2R(\alpha+k\beta)}{1+k^2-2R(1+\gamma)(\alpha+k\beta)}.\label{pureprob}
\end{align}

To derive the sufficient condition in statement (i) of the Theorem, we need to ensure that Eq.~\eqref{eq:suff_prob_condition} holds for all $k$.
This can be transformed to an inequality independent of $k$,  by choosing quantity on the left side, $p(k, \mathsf{b})$, to be equal to a constant larger than the maximum value of the right side, $\max_k p(k, \mathsf{b}_k)$.
Conversely, to derive the necessary condition in part (ii) of the Theorem, we require $p(k, \mathsf{b}) > \min_k p(k, \mathsf{b}_k)$ for all $k$.
Similarly, we can achieve such an inequality by considering the set of $\mathsf{b}$ for which $p(k, \mathsf{b})$ is a constant larger than $\min_k p(k, \mathsf{b}_k)$.
To this end, we find the stationary points of $p(k, \mathsf{b}_k)$.

When $\beta\neq 0$, these occur for the slopes $k_1, k_2$, where
\begin{align}
k_1=\frac{-\alpha-\sqrt{\alpha^2+\beta^2}}{\beta},\quad
k_2=\frac{-\alpha+\sqrt{\alpha^2+\beta^2}}{\beta}.
\end{align}
Denoting the extrema by $p^{+} \coloneqq p(k_1, \mathsf{b}_k)$ and $p^{-} \coloneqq  p(k_2, \mathsf{b}_k)$, we compute
\begin{align}\label{probpn}
&p^{\pm}=\frac{1-\gamma -R^2\beta ^2 (1+\gamma) -R\alpha  \left(2-\gamma ^2\right) \pm R\gamma ^2 \sqrt{\alpha ^2+\beta ^2}}{1-R(1+\gamma)  \left(2 \alpha +R\beta ^2 (1+\gamma)\right)}.
\end{align}
In order to compare $p^{+}$ and $p^{-}$, we calculate
\begin{align}
p^{+}-p^{-}
=\frac{2R\gamma ^2 \sqrt{\alpha ^2+\beta ^2}}{1-R(1+\gamma)  \left(2 \alpha +R\beta ^2 (1+\gamma)\right)}
=\frac{2(\alpha^2+\beta^2)^\frac{3}{2}}{R\beta^2\gamma^2}u_{\mathsf{b}_{k_1}}u_{\mathsf{b}_{k_2}}>0,
\end{align}
where $u_{\mathsf{b}_{k_1}}>0$, $u_{\mathsf{b}_{k_2}}>0$, because for any $k$, $u_{\mathsf{b}_k}>0$. 
Hence, we can deduce that $p^{+}$ is the maximum value of $p(k, \mathsf{b}_k)$ with respect to $k$, and $p^{-}$ is its minimum value. 
Therefore,
\begin{align}
p_{\text{max}}=p^{+},\quad p_{\text{min}}=p^{-}.
\end{align}

 For the case where $\beta= 0$, Eq.~\eqref{pureprob} reduces to
\begin{align}\label{probzero}
p(k, \mathsf{b}_k)
=\frac{(1+k^2)(1-\gamma)-2R\alpha}{1+k^2-2R\alpha(1+\gamma)}.
\end{align}  
The stationary point occurs when $k=0$, and so  
\begin{align}\label{pkzero}
p(0, \mathsf{b}_k) 
=\frac{1-\gamma-2R\alpha}{1-2R\alpha(1+\gamma)}=:p^0.
\end{align} 
In order to judge whether this is a  maximum or minimum value of $p(k, \mathsf{b}_k)$,
we calculate the value of the second derivative of Eq.~\eqref{probzero} at $k=0$, 
\begin{align}\label{secdrt}
\pdv[2]{p(k, \mathsf{b}_k)}{k} \Bigr|_{k=0}  =\frac{4 \alpha  \gamma ^2 R}{(1-2 \alpha  (\gamma +1) R)^2}.
\end{align} 
 There are three relevant cases, based on the sign of $\alpha$.
The first is where $\alpha>0$, and so Eq.~\eqref{secdrt} is positive.
Then, 
\begin{align}\label{appp1}
p_{\text{min}}=p^{0},\quad
p_{\text{max}}=\lim_{k\to\infty}p(k, \mathsf{b}_k)=1-\gamma.
\end{align}
If $\alpha<0$, then Eq.~\eqref{secdrt} is negative, and we have
\begin{align}\label{appp2}
p_{\text{max}}=p^{0},\quad
p_{\text{min}}=1-\gamma.
\end{align}
Finally, if $\alpha=0$, then Eq.~\eqref{probzero} reduces to $1-\gamma$. 
This means that
\begin{align}\label{appp3}
p_{\text{max}}=p_{\text{min}}=1-\gamma.
\end{align}

All of the above considerations are in a particular plane $P_{\theta\phi}$.
Restoring these labels, we write $p_{\text{max}}$ and $p_{\text{min}}$ as $p_{\text{max}}^{\theta\phi}$ and $p_{\text{min}}^{\theta\phi}$, respectively. 
Considering all the planes $P_{\theta\phi}$, we arrive at the statement of the Theorem, by defining the probabilities which appear there as
\begin{align}\label{maxminp}
p_{\text{max}}^{\mathcal{E}} \coloneqq \max_{\theta\phi}p_{\text{max}}^{\theta\phi},\quad
p_{\text{min}}^{\mathcal{E}} \coloneqq \min_{\theta\phi}p_{\text{min}}^{\theta\phi}.
\end{align}
\end{proof}

There is an interesting geometric interpretation of Eq.~\eqref{eq:suff_prob_condition} and its converse.
For any plane $P_{\theta\phi}$, the set of points $\mathsf{b}$ for which $p(k, \mathsf{b})$ is equal to the constant $p_{\text{max}}$ are given by an 
affine ``shrinking'' transformation of the ellipse $\mathsf{E}$.
In terms of conics, this transformation is given by $A=(F_A^{-1})^\top E F_A^{-1}$, where
\begin{align}
{F}_A^{-1}=\text{diag}\left(1,1,1-p^{+}\right).
\end{align}
Analogously, the set of $\mathsf{b}$'s for which $p(k, \mathsf{b})$ is equal to the constant $p_{\text{min}}$ is equivalent to the affine transformation $E\mapsto D=(F_D^{-1})^\top E F_D^{-1}$ where
\begin{align}
{F}_D^{-1}=\text{diag}\left(1,1,1-p^{-}\right).
\end{align}
These two \emph{shrunken} ellipses are shown in Fig.~\ref{five} as the ellipses $\mathsf{A}: \{\mathsf{a}_k|\overrightarrow{\mathsf{pa}}_k=(1-p_{\text{max}})\overrightarrow{\mathsf{pe}}_k\}$ and $\mathsf{D}: \{\mathsf{d}_k|\overrightarrow{\mathsf{pd}}_k=(1-p_{\text{min}})\overrightarrow{\mathsf{pe}}_k\}$, where $\mathsf{a}_k$ and $\mathsf{d}_k$ are points on the ellipse $\mathsf{A}$ and $\mathsf{D}$ corresponding to point $\mathsf{e}_k$, respectively.
In terms of the Theorem, for any $\textsf{b}$, consider the one-parameter family of planes ${\cal P}_b=\{P_\delta| \delta \in [ 0,2\pi)\}$, containing the line $\mathsf{pb}$,  where $P_{\delta}$ is the plane rotated by the angle $\delta$ around the line $\mathsf{pb}$. 
If $\mathsf{b}$ is interior to $\mathsf{A}$ in $ P_\mathsf{\delta}$ for all $\delta$, we know it satisfies the sufficient condition for steering in part (i).
 Moreover, if $\mathsf{b}$ is interior to $\mathsf{D} $ in $ P_\mathsf{\delta}$ for some $\delta$, it satisfies the necessary condition, as per statement (ii).

\subsection{Examples} 
\label{sec:final_examples}

Now we apply Theorem~\ref{theorem2} to two families of states. The first is the family of tangent X-states previously considered in Sec~\ref{sec:tangent_x_states}. The second is a new family, which we call {\em tangent spheroid states}, which includes the family of the tangent sphere states, Sec.~\ref{sec:tangent_spheres}, 
as a special case. (Note that the canonical obese states, Sec.~\ref{sec:can_obese}, are a special case of the tangent spheroid states.)

We can deal with tangent X-states in a single paragraph. Recall from Sec~\ref{sec:tangent_x_states} that this is a 4-parameter family of states, of which 3 parameters define the ellipsoid ${\cal E}$ and one the Bloch vector of Bob's reduced state, as this is guaranteed to be on the $u$-axis. 
This means that we only need to consider the planes $\{P_{\theta}, \theta\in [0,\pi/2]\}$ containing the $u$-axis, hence $R=1$ and $k=0$. 
Using Eqs.~\eqref{xabc}, 
\eqref{pkzero}, and \eqref{maxminp}, we derive
\begin{align}
& p_{\text{max}}^{\mathcal{E},\text{X}}=\frac{m(1-m)}{m(1-m)+n^2_y},\\
& p_{\text{min}}^{\mathcal{E},\text{X}}=\frac{m(1-m)}{m(1-m)+n^2_x},
\end{align}
 and we remind the reader that $n_x$, $n_y$, and $m$ are three semiaxes of the steering ellipsoid with $n_x>n_y$. 
That is, from Theorem~\ref{theorem2}, we have a sufficient condition and a necessary condition, for Alice to demonstrate EPR-steering in the Scenario, by any or some second measurement respectively, on the probability that Bob's one pure state appears in one of his ensembles,  in terms of his steering ellipsoid ${\cal E}$.

The family of tangent spheroid states is a more interesting application of Theorem~\ref{theorem2}. 
For this class, the steering ellipsoid is a spheroid with two equal and one distinct semiaxes, and touches the surface of Bob's Bloch ball at the endpoint of the distinct semiaxis. 
 There are no restrictions on the location of $\mathsf{b}$ inside this spheroid. Thus, this is a 5-parameter family of states, but the steering ellipsoid ${\cal E}$ is specified by just two. (Moreover, by the rotational symmetry of ${\cal E}$ around the $u$-axis, one parameter in $\mathsf{b}$ is uninteresting, so it could always be defined as a 4-parameter family of states.) 
As shown in Fig.~\ref{spheroid} (a), 
 we begin by defining a $(u,v,w)$ coordinate system with the origin at $\mathsf{p}$.
The equation of the spheroid can be written as
\begin{align}
\frac{(u-m)^2}{m^2}+\frac{v^2}{n^2}+\frac{w^2}{n^2}=1, 
\end{align}
 and the equation describing Bloch sphere is
\begin{align}
(u-1)^2+v^2+w^2=1.
\end{align}

\begin{figure}[h]\centering
\includegraphics[angle=0,width=0.35\linewidth]{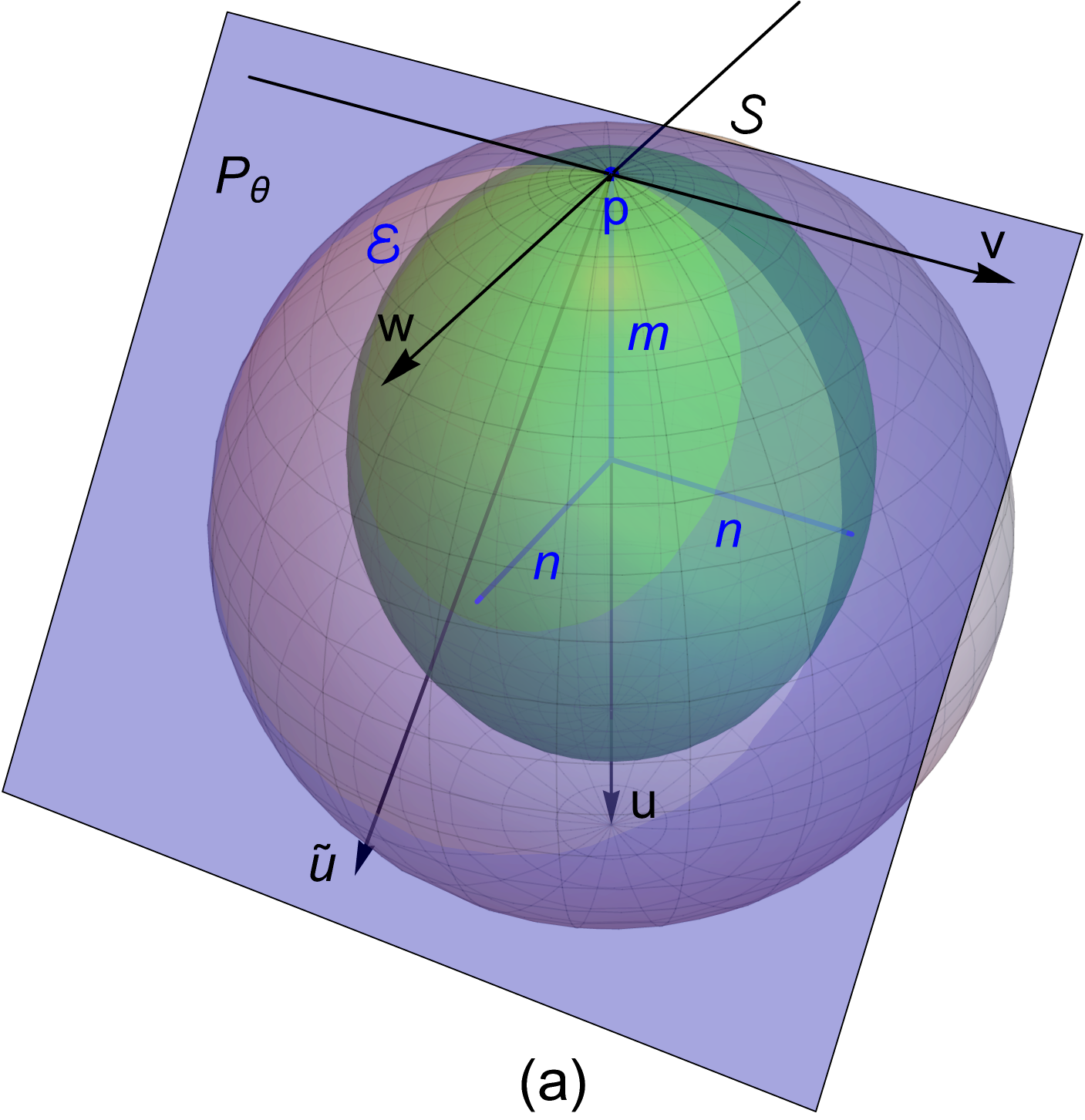}
\hspace{17pt}
\includegraphics[angle=0,width=0.31\linewidth]{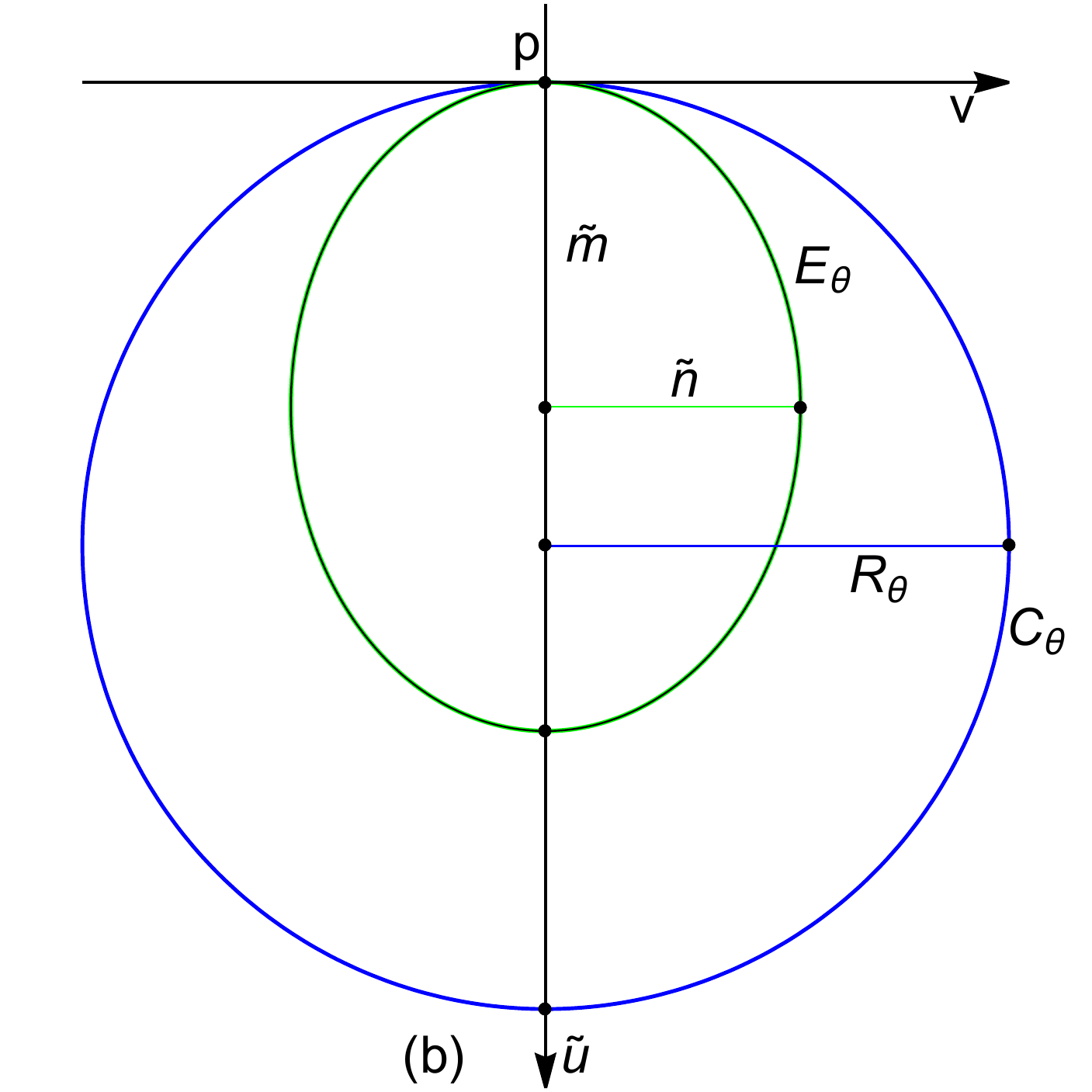}
\vspace{-10pt}
\caption{The cross sections of a plane $P_{\theta}$ with the Bloch sphere ${\cal S}$ and steering ellipsoid (spheroid) ${\cal E}$.
 These cross sections are a circle $\mathsf{C_{\theta}}$ and ellipse $\mathsf{E_{\theta}}$, respectively. 
For $P_\theta$, we define a $(\tilde{u}, v)$ Cartesian coordinate system originating at $\mathsf{p}$, where the $\tilde{u}$-axis is defined by the intersection of the planes $uw$ and $P_\theta$. }
\label{spheroid}
\end{figure}

Due to the symmetry of the spheroid under rotations about the $u$-axis, we need to consider only a single-parameter set of planes $\mathbb{P}=\{P_\theta \mid \theta\in[0,\pi/2]\}$, which pass through $\mathsf{p}$, defined by the angle of the plane with the $u$-axis.
 Each plane cuts both the spheroid and Bloch sphere, and these intersections are an ellipse $\mathsf{E}_\theta$ and a circle $\mathsf{C}_\theta$.
 Since we can choose all planes in this set to include the $v$-axis, we can describe points in any $P_\theta$ according to a $(\tilde{u}, v)$ Cartesian coordinate system with its origin at $\mathsf{p}$.
The relevant geometry in a particular plane is illustrated in Fig~\ref{spheroid}~(b). 
The equation of the ellipse $\mathsf{E}_\theta$ can be expressed in the $(\tilde{u},v)$ coordinate system as
\begin{align}
\frac{(\tilde{u}-\tilde{m})^2}{\tilde{m}^2}+\frac{v^2}{\tilde{n}^2}=1,
\end{align}
where the $\tilde{u}$-axis is the intersection of the planes $uw$ and $P_\theta$, and 
\begin{align}\label{mntheta}
\tilde{m} \coloneqq  \frac{m n^2 \cos (\theta )}{m^2 \sin ^2(\theta )+n^2 \cos ^2(\theta )},\quad
\tilde{n} \coloneqq  \frac{n^2 \cos (\theta )}{\sqrt{m^2 \sin ^2(\theta )+n^2 \cos ^2(\theta )}}.
\end{align}
The equation of the circle $\mathsf{C}_\theta$ is
\begin{align}
(\tilde{u}-R_\theta)^2+v^2=R_\theta^2,
\end{align}
where $R_\theta=\cos(\theta)$. 
 We can now directly apply the methodology used in the proof of  Theorem~\ref{theorem2}. 
 The projective transformation is defined by 
\begin{align}
\begin{pmatrix}
 \alpha  \\
 \beta  \\
 \gamma
\end{pmatrix}
=\frac{1}{2\tilde{m}^{2}\cos(\theta)}\begin{pmatrix}
 \tilde{m}^2-\tilde{n}^2\\
 0  \\
2 \tilde{m}\tilde{n}^2
\end{pmatrix}.
\end{align}

There are three cases to compute, based on the relative sizes of the semiaxes of $\mathsf{E}_\theta$.
First, if $\tilde{m}>\tilde{n}$, using Eq.~\eqref{appp1} and \eqref{maxminp}, we find
\begin{align}
p_{\text{min}}^{\mathcal{E},\text{S}}=\frac{(1-m) m}{m(1-m)+n^2},\quad
p_{\text{max}}^{\mathcal{E},\text{S}}=1-\frac{n^2}{m}.
\end{align}
 However, if  $\tilde{m}<\tilde{n}$, using Eq.~\eqref{appp2} and \eqref{maxminp}, we have
\begin{align}
p_{\text{max}}^{\mathcal{E},\text{S}}=\frac{(1-m) m}{m(1-m)+n^2},\quad
p_{\text{min}}^{\mathcal{E},\text{S}}=1-\frac{n^2}{m}.
\end{align}
 Finally, when  $\tilde{m}=\tilde{n}$, then according to Eq.~\eqref{mntheta}, $m=n$. Using Eq.~\eqref{appp3} and \eqref{maxminp}, we obtain
\begin{align}\label{sphpp}
p_{\text{max}}^{\mathcal{E},\text{S}}=p_{\text{min}}^{\mathcal{E},\text{S}}=1-m.
\end{align}
In this last case, $m=n$, the tangent spheroid states reduce to the case of tangent sphere states, discussed in Section~\ref{sec:tangent_spheres}.
There, in Eq.~\eqref{tssp}  we obtained the necessary and sufficient condition $p_\mathsf{p} > 1-r$ to demonstrate EPR-steering.
This condition is in fact equivalent to Eq.~\eqref{sphpp}. \blk
From the point of view of Theorem~\ref{theorem2}, this is because $p_{\rm max}^{{\cal E}}$ is equal to $p_{\rm min}^{{\cal E}}$ for all tangent sphere states. 
That is, Theorem~\ref{theorem2} is tight for tangent sphere states.

\section{Conclusion}
\label{sec:Conclusion}

In this paper, we considered  
the simplest scenario for bipartite steering---when Alice attempts to steer Bob by making two dichotomic measurements---under the condition that for one and only one of the outcomes of one and only one of those measurements, Bob is steered to a pure state. We derived some general results using elementary planar geometry before restricting to the case where Alice's and Bob's systems are qubits, in which steering is characterized by Bob's Bloch vector $\mathsf{b}$ and his steering ellipsoid ${\cal E}$. 
Using tools from projective geometry, we found a simple characterization of steerability in this context (Theorem~\ref{theorem}). 
We applied this to three families of entangled states, finding both necessary conditions and sufficient conditions (which for one family coincide) on $\mathsf{b}$ and ${\cal E}$ for EPR-steerability. 
Finally, we considered what one can say if one knows only the ellipsoid ${\cal E}$, which is of non-zero volume and touches the Bloch sphere exactly once.
We showed that there always exists two non-trivial bounds on the probability that Bob is steered to the pure state which provides, respectively,  \emph{sufficient} and \emph{necessary} conditions, for Alice to demonstrate EPR-steering using the pure-state-inducing measurement and {\em any} or {\em some} (respectively) other  dichotomic projective measurement (Theorem~\ref{theorem2}).

We conclude with a discussion of open questions. 
There are a few interesting directions for generalizing the link between steering ellipsoid geometry, projective geometry, and EPR-steerability. 
The first involves going beyond the above scenario to consider assemblages with three or more ensembles in a single plane intersecting the Bloch sphere, one of which contains a pure state.
In this case, is there a simple geometric characterization of steerability (akin to Lemma~\ref{lemma2}) which has interesting properties under the image of a projective transformation? 
The second direction goes beyond assemblages of finite ensembles, to consider directly the geometry of general steering ellipsoids touching Bob's Bloch sphere at one point.
Previously, steering ellipsoid geometry was used to exactly characterize the steerability of mixtures of Bell states under all projective measurements~\cite{Jev15, Ngu16}.
Can we gain insight from projective geometry, when one of these projectors steers to a pure state? 
We conjecture, in light of Theorem~\ref{theorem2}, that there exists a smooth convex hull---which may be an ellipsoid, or something more exotic---contained inside the steering ellipsoid which also touches at the pure state, such that if the Bob's reduced state is strictly inside it,
Alice can be able to demonstrate EPR-steering. 
Furthermore, can the method of projective geometry be applied to higher dimensions in a generalized manner? 
Finally, it is possible that projective geometry could be used to derive useful conditions for EPR-steering even in cases without pure steered states.
In our proofs, we constructed special types of projective transformations, called a planar homology, which had the pure steered state at the centre of its projection in all planes of the Bloch sphere.
Are there other types of projective transformations useful to geometrically characterize steerability?

\section{Acknowledgement}

We thank Richard Gill and Bas Edixhoven for pointing out the connection to projective geometry in Lemma \ref{lemma3}. 
Qiu-Cheng Song acknowledges support by a cotutelle Scholarship from Griffith University and University of Chinese Academy of Sciences. 
This work was supported by the ARC Centre of Excellence for Quantum Computation and Communication Technology (CQC2T), project number CE170100012.

\section*{References}
\bibliographystyle{vancouver}
\bibliography{main.bib}
\end{document}